\documentclass[11pt]{article}%
\usepackage{amsfonts}
\usepackage{amsmath}
\usepackage{amssymb}
\usepackage{graphicx}%
\providecommand{\U}[1]{\protect\rule{.1in}{.1in}}
\newtheorem{theorem}{Theorem}
\newtheorem{acknowledgement}[theorem]{Acknowledgement}

\newtheorem{corollary}[theorem]{Corollary}

\newtheorem{definition}[theorem]{Definition}
\newtheorem{example}[theorem]{Example}

\newtheorem{lemma}[theorem]{Lemma}

\newtheorem{proposition}[theorem]{Proposition}
\newtheorem{remark}[theorem]{Remark}

\newenvironment{proof}[1][Proof]{\noindent\textbf{#1.} }{\ \rule{0.5em}{0.5em}}
\begin{document}

\title{Paths of Canonical Transformations and their Quantization}
\author{Maurice A. de Gosson\thanks{maurice.de.gosson@univie.ac.at}\\University of Vienna\\Institute of Mathematics (NuHAG)\\Oskar-Morgenstern-Platz 1, 1090 Vienna}
\maketitle
\tableofcontents

\begin{abstract}
In their simplest formulations, classical dynamics is the study of Hamiltonian
flows and quantum mechanics that of propagators. Both are linked, and emerge
from the datum of a single classical concept, the Hamiltonian function. We
study and emphasize the analogies between Hamiltonian flows and quantum
propagators; this allows us to verify G. Mackey's observation that quantum
mechanics (in its Weyl formulation) is a refinement of Hamiltonian mechanics.
We discuss in detail the metaplectic representation, which very explicitly
shows the close relationship between classical mechanics and quantum
mechanics, the latter emerging from the first by lifting Hamiltonian flows to
the double covering of the symplectic group. We also give explicit formulas
for the factorization of Hamiltonian flows into simpler flows, and prove a
quantum counterpart of these results.

\end{abstract}

\section{Introduction}

Hamiltonian and quantum mechanics share a common background: they are both
based on the datum of a function $H$ (the Hamiltonian) which is a
generalization of the notion of energy. While this function is taken literally
in Hamiltonian mechanics as defining a vector field $X_{H}$ on phase space
which gives rise to the \textquotedblleft Hamiltonian flow\textquotedblright%
\ $(f_{t}^{H})_{t}$, in quantum mechanics one associates with $H$ a
differential, or pseudo-differential, operator generating a group of unitary
transforms $(U_{t}^{H})_{t}$ satisfying the Schr\"{o}dinger equation. When the
Hamiltonian function is a quadratic form in the phase space variables, both
$(f_{t}^{H})_{t}$ and $(U_{t}^{H})_{t}$ are easily linked: in this case the
$f_{t}^{H}$ are just linear canonical transformations, and the $U_{t}^{H}$ are
(up to an easily determined phase factor) the metaplectic operators obtained
by lifting the flow $(f_{t}^{H})_{t}$ to the metaplectic group. We will, among
other things, extend this similarity to arbitrary Hamiltonian flows and
propagators, thus producing an extension of the metaplectic representation.
One of the aims of the present paper is actually to highlight the similarities
between classical mechanics (in its Hamiltonian formulation), and quantum
mechanics: the latter is seen as emerging from the first. In fact, all which
lacks is a physical motivation for the introduction of Planck's constant $h$.
We will not discuss this very important -- and indeed difficult! -- problem
here; for all practical purposes we will view Planck's constant as a scaling
factor, and we will not discuss the physical content of quantum mechanics (in
the Schr\"{o}dinger picture, which is the only one considered here): there is
a vast literature on the subject, and it seems that no real consensus has yet
been reached.

A point of terminology: we will indifferently use the words \textquotedblleft
canonical transformation\textquotedblright\ or \textquotedblleft
symplectomorphism\textquotedblright. While in some texts both notions are not
quite equivalent, \textquotedblleft canonical transformation\textquotedblright%
\ sometimes having meaning any diffeomorphism preserving the form of
Hamilton's equations, we will adhere to the stricter definition following
which a canonical transformation is a diffeomorphism of phase space whose
Jacobian matrix is symplectic.

Here are some highlights of this paper:

\begin{itemize}
\item In Section \ref{sec1} we study Hamiltonian flows from the symplectic
point of view, and prove factorization formulas. For instance we show in
Proposition \ref{prop6} that if a Hamiltonian $H$ is split into a sum of two
independent Hamiltonians $H_{0}+H_{1}$ then the flow $(f_{t}^{H})_{t}$
determined by $H$ can be written as the product $f_{t}^{H_{0}}f_{t}^{H_{1}%
^{t}}$ where $H_{1}^{t}(z,t)=H_{1}(f_{t}^{H_{0}}(z),t)$. These constructions
allow us to define the group $\operatorname*{Ham}(n)$ of all Hamiltonian
symplectomorphisms associated with compactly supported Hamiltonians. We also
discuss a few special properties of Hamiltonian flows (and canonical
transformations) such as the theorems of Katok and Gromov.

\item In Section \ref{sec2} we give a rigorous review of the metaplectic group
and of its inhomogeneous extension. This section is in a sense pivotal,
because it is the gate to quantum mechanics, but a gate only using
\textit{classical} concepts (the path lifting property for covering groups).
It shows, in a sense, that quantum mechanics emerges from classical
(Hamiltonian) mechanics by \textquotedblleft passing to the covering
group\textquotedblright\ (there is a certain similarity with the passage from
the rotations group to the spin group). We discuss the notion of Feichtinger
algebra which is the most natural domain for the metaplectic group and its
inhomogeneous extension. We finally shortly review the theory of the Weyl
symbol of a metaplectic operator, following previous work of ours.

\item In Section \ref{sec3}, after having discussed quantization in general,
we define the notion of \textit{quantum isotopy}: a quantum isotopy is a
continuous path $(U_{t})_{t}$ of unitary operators satisfying a certain
differentiability condition which allows us to define a canonical self-adjoint
operator by the formula $\widehat{H}=i\hbar\left(  \frac{d}{dt}U_{t}\right)
U_{t}^{-1}$. This operator coincides with the infinitesimal generator obtained
by Stone's theorem when $(U_{t})_{t}$ is a strongly continuous one-parameter
group. We thereafter prove in Proposition \ref{propho} a quantum analogue of
the factorization result of \ref{prop6}: we have $U_{t}^{H_{0}+H_{1}}%
=U_{t}^{H_{0}}U_{t}^{H_{1}\circ S_{t}}$ if $H_{0}$ is quadratic.
\end{itemize}

\paragraph{Notation}

We will identify $T^{\ast}\mathbb{R}^{n}=\mathbb{R}^{n}\times(\mathbb{R}%
^{n})^{\ast}$ with $\mathbb{R}^{2n}$ and denote by $\sigma$ the standard
symplectic 2-form $\sum_{j=1}^{n}dp_{j}\wedge dx_{j}$, \textit{i.e.}
$\sigma(z,z^{\prime})=Jz\cdot z^{\prime}$ where $z=(x,p)$, $z^{\prime
}=(x^{\prime},p^{\prime})$ and $J=%
\begin{pmatrix}
0 & I\\
-I & 0
\end{pmatrix}
$. The scalar product of two vectors $u,v\in\mathbb{R}^{m}$ is written $u\cdot
v$; if $A$ is a symmetric $m\times m$ matrix we will often write $Au\cdot
u=Au^{2}$. We will denote by $\hbar$ a positive parameter which is identified
in physics with Planck's constant $h$ divided by $2\pi$. The Schwartz space of
complex functions on $\mathbb{R}^{m}$ decreasing at infinity, together with
their derivatives, faster than the inverse of any polynomial is denoted by
$\mathcal{S}(\mathbb{R}^{m})$; its dual $\mathcal{S}^{\prime}(\mathbb{R}^{m})$
is the space of tempered distributions. The scalar product on $L^{2}%
(\mathbb{R}^{m})$ is written $\langle\cdot|\cdot\rangle$.

\section{Symplectic and Hamiltonian Isotopies\label{sec1}}

The symplectic group $\operatorname*{Sp}(n)$ is the group of all (linear)
automorphisms of the symplectic space $(\mathbb{R}^{2n},\sigma)$:
$S\in\operatorname*{Sp}(n)$ if and only if $S\in GL(2n,\mathbb{R})$ and
$\sigma(Sz,Sz^{\prime})=\sigma(z,z^{\prime})$ for all $(z,z^{\prime}%
)\in\mathbb{R}^{2n}\times\mathbb{R}^{2n}$ (for a detailed exposition, see
\cite{Birk}). A diffeomorphism $f$ of $\mathbb{R}^{2n}$ is a symplectomorphism
if its Jacobian matrix at every point is symplectic: $Df(z)\in
\operatorname*{Sp}(n)$ for every $z\in\mathbb{R}^{2n}$. Equivalently,
$(f_{t}^{H})^{\ast}\sigma=\sigma$, identifying $\sigma$ with the 2-form
$\sum_{j=1}^{n}dp_{j}\wedge dx_{j}$. The symplectomorphisms of $(\mathbb{R}%
^{2n},\sigma)$ form a group $\operatorname*{Symp}(n)$.

\subsection{Hamiltonian vector fields and flows}

\subsubsection{Hamilton's equations}

Let $H\in C^{j}(\mathbb{R}^{2n}\times\mathbb{R})$, $j\geq2$, be a real-valued
function; we will call $H$ a \emph{Hamiltonian function}. The associated
Hamilton equations with initial data $z^{\prime}$ at time $t^{\prime}$ are%
\begin{equation}
\dot{x}=\partial_{p}H(x,p,t)\text{ \ , \ }\dot{p}=-\partial_{x}H(x,p,t)
\label{hameq}%
\end{equation}
or, using the collective notation $z=(x,p)$,
\begin{equation}
\dot{z}(t)=X_{H}(z(t),t)\text{ \ , \ }z(t^{\prime})=z^{\prime} \label{hameq1}%
\end{equation}
where $X_{H}=J\partial_{z}H$ is the Hamilton vector field (strictly speaking
it is not a true vector field when $H$ depends on $t$). Existence and
uniqueness results for Hamilton's equations abound in the literature (see e.g.
\cite{AM,AMR,zehnder} and the multiple references therein). The study of these
properties are actually mostly a branch of the theory (local, and global) of
dynamical systems (= systems of ordinary differential equations) where the
notion of vector fields and their integral curves play the primordial role.
The main result we will (implicitly) use is the following local classical
existence and uniqueness property: Let $X:\mathbb{R}^{2n}\longrightarrow
\mathbb{R}^{2n}$ be a vector field (Hamiltonian, or not) of class $C^{j}$,
$j\geq1$ (it is hence, in particular, locally Lipschitz continuous). For every
$z_{0}\in\mathbb{R}^{2n}$ the system $\dot{z}=X(z)$ there exists an open ball
$B^{2n}(z_{0},r)\subset\mathbb{R}^{2n}$ and an $\varepsilon>0$ such that this
system has a unique solution $t\longmapsto z(t)$ with $z(0)=z_{0}$ which is
$C^{j}$ and defined for every $t\in I_{\varepsilon}=[-\varepsilon
,\varepsilon]$. Moreover, by uniqueness, every such solution can be continued
onto a so-called maximal interval $I_{T}=[-T,T]$. The same result holds for
time-dependent vector fields $X_{t}$, which are families of \textquotedblleft
true\textquotedblright\ vector fields depending in a $C^{1}$ fashion on time
$t$. Transposed to the case of Hamiltonian systems, this means that it is
sufficient, to have a local existence and uniqueness statement to assume that
$H\in C^{j}(\mathbb{R}^{2n}\times\mathbb{R})$, $j\geq2$ as we did. In
practice, we will assume that the maximal interval $I_{T}=[-T,T]$ is the same
for every initial point $z_{0}$, and whenever this simplifies statements, that
$T=+\infty$ (which is the case if the Hamiltonian function is constant outside
some compact set in $\mathbb{R}^{2n}$, or, more generally, if the vector field
$X_{H}$ is complete; see Abraham and Marsden \cite{AM} for details).

We have the following transformation formula for Hamiltonian vector fields: if
$g\in\operatorname*{Symp}(n)$ then
\begin{equation}
X_{H\circ g}(z)=[Dg(z)]^{-1}X_{H}(g(z)). \label{XHg}%
\end{equation}
This formula can be written more concisely as $X_{g_{\ast}H}=g_{\ast}X_{H}$
where $g_{\ast}$ is the operation of \textquotedblleft pushing
forward\textquotedblright.

\subsubsection{Hamiltonian flows}

Assuming existence and uniqueness of the solution for every choice of
$(z^{\prime},t^{\prime})$ the time-dependent flow $(f_{t,t^{\prime}}^{H})$ is
the family of mappings $\mathbb{R}^{2n}\longrightarrow\mathbb{R}^{2n}$ which
associates to every initial $z^{\prime}$ the value $z(t)=f_{t,t^{\prime}}%
^{H}(z^{\prime})$ of the solution of (\ref{hameq1}). We will write $f_{t}%
^{H}=f_{t,0}^{H}$ and we have
\begin{equation}
f_{t,t^{\prime}}^{H}=f_{t,0}^{H}\left(  f_{t^{\prime},0}^{H}\right)
^{-1}=f_{t}^{H}\left(  f_{t^{\prime}}^{H}\right)  ^{-1} \label{ftt'}%
\end{equation}
\ and the $f_{t,t^{\prime}}^{H}$ satisfy the groupoid property
\begin{equation}
f_{t,t^{\prime}}^{H}f_{t^{\prime},t^{\prime\prime}}^{H}=f_{t,t^{\prime\prime}%
}^{H}\text{ \ , \ }f_{t,t}^{H}=I_{\mathrm{d}} \label{flow2}%
\end{equation}
for all $t$, $t^{\prime}$ and $t^{\prime\prime}$. It follows from the first
formula (\ref{flow2}) that we have $(f_{t,t^{\prime}}^{H})^{-1}=f_{t^{\prime
},t}^{H}$.

The mappings $f_{t}^{H}$ (and hence the mappings $f_{t,t^{\prime}}^{H}$) are
symplectomorphisms, that is $(f_{t}^{H})^{\ast}\sigma=\sigma$. To see this, it
suffices to notice that $i_{X_{H}}\sigma=dH$ and hence, in view of Cartan's
formula for the Lie derivative \cite{AM},%
\[
L_{X_{H}}\sigma=d(i_{X_{H}}\sigma)+i_{X_{H}}d\sigma=0
\]
which implies that
\[
\frac{d}{dt}(f_{t}^{H})^{\ast}\sigma=f_{t}^{H})^{\ast}L_{X_{H}}\sigma=0
\]
and hence $(f_{t}^{H})^{\ast}\sigma=(f_{0}^{H})^{\ast}\sigma=\sigma$ (for a
conceptually simpler, but longer, proof using Jacobian matrices, see
\cite{Birk}).

It follows from the transformation formula (\ref{XHg}) for Hamiltonian vector
fields that the flows $(f_{t}^{H})_{t}$ and $(f_{t}^{H\circ g})_{t}$ are
conjugate:%
\begin{equation}
f_{t}^{H\circ g}=g^{-1}f_{t}^{H}g; \label{FHg}%
\end{equation}
this is the property of \textquotedblleft covariance of Hamilton's equations
under canonical transformations\textquotedblright\ familiar from physics (see
Arnol'd \cite{Arnold}, Hofer and Zehnder \cite{HZ}).

\subsubsection{Special features of Hamiltonian flows}

Hamiltonian flows are volume preserving (because the Jacobian determinant of a
canonical transformation is always equal to one). They however enjoy several
unusual properties which makes them very different from arbitrary
volume-preserving diffeomorphisms. For instance, Gromov proved in 1985 that no
Hamiltonian flow (or more generally, no canonical transformation) can embed a
phase space ball $B^{2n}(R)$ with radius $R$ inside a cylinder $Z_{j}^{2n}(r)$
based on a plane $x_{j},p_{j}$ of conjugate coordinates if its radius $r$ is
smaller than $R$. This unexpected property, sometimes dubbed the
\textquotedblleft principle of the symplectic camel\textquotedblright, has led
to various developments in symplectic topology. The principle of the
symplectic camel has the following consequence: let us project orthogonally a
ball with radius $R$ on a plane of conjugate coordinates $x_{j},p_{j}$, we
then obtain a circle with area $\pi R^{2}$. If we now deform the ball using a
Hamiltonian flow, it will deform, but the area of the orthogonal projection of
this deformed ball will never decrease below its initial value $\pi R^{2}$.
This result is a \emph{classical version} of the uncertainty principle as we
have explained in de Gosson \cite{Birk,FP,stat} and de Gosson and Luef
\cite{golu09}. We mention that Gromov's theorem in addition allows to define a
new class of symplectic invariants, called \emph{symplectic capacities}, which
seem to play an important role in accelerator physics\emph{ }(see Dragt
\cite{Dragt} and Erdelyi \cite{er11}). The properties above are of a
topological nature, and show that general volume preserving mapping cannot be
approximated in the $C^{0}$ topology by canonical transformations. On the
other hand, Katok \cite{Katok} has showed that given two subsets $\Omega$ and
$\Omega^{\prime}$ with same volume, then for every $\varepsilon>0$ there
exists a canonical transformation $f$ such that $\operatorname*{Vol}%
(f(\Omega)\setminus\Omega^{\prime})<\varepsilon$. Thus, an arbitrarily large
part of $\Omega$ \ can be symplectically embedded inside $\Omega^{\prime}$. A
third property (which is however shared by other groups of diffeomorphisms) is
$N$-transitivity. Let us introduce the following terminology: given a group
$G$ acting on a set $M$ we say that this action is $N$-transitive if given two
arbitrary sets $\{x_{1},...,x_{N}\}$ and $\{y_{1},...,y_{N}\}$ of points of
$M$ there exists $g\in G$ such that $g(x_{i})=y_{i}$ for all $i=1,2,..,N$. For
instance, it is well-known that $\operatorname*{Diff}(\mathbb{R}^{m})$ (the
group of diffeomorphisms of $\mathbb{R}^{m}$) acts $N$-transitively on
$\mathbb{R}^{m}$ for every $N$. Now, a deep theorem of Boothby \cite{Boothby}
(also see \cite{ha66,mivi94}) says that the action of $\operatorname*{Ham}(n)$
on $\mathbb{R}^{2n}$ is $N$-transitive for every integer $N$. In fact, given
two arbitrary sets $\{z_{1},...,z_{N}\}$ and $\{z_{1}^{\prime},...,z_{N}%
^{\prime}\}$ of $N$ distinct points in $\mathbb{R}^{2n}$, there exists a
Hamiltonian function $H$ such that $z_{j}^{\prime}=f_{1}^{H}(z_{j})$ for every
$j=1,...,N$. This property certainly has applications (for instance to
celestial mechanics) which have not been explored yet.

\subsection{Paths of canonical transformations}

A striking result is that any continuously differentiable path of canonical
transformations passing through the identity can be viewed as the Hamiltonian
flow determined by some (usually time-dependent) Hamiltonian. We will prove
this in Proposition \ref{prop5}, and draw interesting consequences about the
group of Hamiltonian symplectomorphisms.

\subsubsection{Symplectic and Hamiltonian isotopies}

We will call a canonical transformation $f$ such that $f=f_{t}^{H}$ for some
Hamiltonian function $H$ and time $t=a$ a \emph{Hamiltonian symplectomorphism}%
. The choice of time $t=a$ in this definition is of course arbitrary, and can
be replaced with any other value different from zero noting that we have
$f=f_{a}^{H_{a}}$ where $H_{a}(z,t)=aH(z,at)$; the usual choice is $a=1$.

\begin{definition}
A symplectic isotopy is a one-parameter family $(f_{t})_{t}$ of
symplectomorphisms depending in a $C^{1}$ fashion on $t\in\mathbb{R}$ and such
that $f_{0}=I_{\mathrm{d}}$. If each $f_{t}$ is a Hamiltonian
symplectomorphism, then $(f_{t})_{t}$ is called a Hamiltonian isotopy.
\end{definition}

It turns out that each symplectic isotopy is a Hamiltonian isotopy, in fact
the flow determined by some time-dependent $H$:

\begin{proposition}
\label{prop5}Let $(f_{t})_{t}$ be a symplectic isotopy. (i) There exists a
Hamiltonian function $H=H(z,t)$ such that $f_{t}=f_{t}^{H}$. (ii) More
precisely, we have $(f_{t})_{t}=(f_{t}^{H})_{t}$ with%
\begin{equation}
H(z,t)=-\int_{0}^{1}\sigma\left(  \tfrac{d}{dt}f_{t}\circ f_{t}^{-1}(\lambda
z),z\right)  d\lambda. \label{hzt}%
\end{equation}
Equivalently:%
\begin{equation}
H(z,t)=-\int_{0}^{1}\sigma\left(  X_{H}(f_{t}^{-1}(\lambda z),z\right)
)d\lambda. \label{hztbis}%
\end{equation}

\end{proposition}

\begin{proof}
(Cf. Wang's \cite{Wang}; for a more conceptual approach using differential
geometry see Banyaga \cite{Banyaga}.) The $f_{t}$ being Hamiltonian
symplectomorphisms we have $Df_{t}(z)^{T}JDf_{t}(z)=J$ for every
$z\in\mathbb{R}^{2n}$. Differentiating both sides of this equality with
respect to $t$ we get, omitting the variable $z$ and writing $\dot{f}%
_{t}=df_{t}/dt$,%
\[
(D\dot{f}_{t})^{T}JDf_{t}+Df_{t}^{T}JD(\dot{f}_{t})=0
\]
and hence
\[
J(D\dot{f}_{t})(Df_{t})^{-1}=\left[  J(D\dot{f}_{t})(D(f_{t})^{-1})\right]
^{T}.
\]
Using the chain rule together with the identity $D(f_{t})^{-1}=D(f_{t}^{-1})$
we have%
\[
D(J\dot{f}_{t}\circ f_{t}^{-1})=\left[  DJ(\dot{f}_{t}\circ f_{t}%
^{-1})\right]  ^{T}%
\]
which shows that the Jacobian of $J(\dot{f}_{t}\circ f_{t}^{-1})$ is
symmetric. It follows from Poincar\'{e}'s lemma that there exists for each
$t\in I$ a function $H_{t}$ such that $J(\dot{f}_{t}\circ f_{t}^{-1}%
)=\partial_{z}H_{t}$ and one verifies that this function is explicitly given
by the formula%
\[
H_{t}(z)=\int_{0}^{1}J\partial_{z}(\dot{f}_{t}\circ f_{t}^{-1})(\lambda
z)\cdot zd\lambda.
\]
Setting $H(z,t)=-H_{t}(z)$ this is precisely formula (\ref{hzt}).
\end{proof}

\subsubsection{Linear and affine flows}

When the Hamiltonian flow is linear, Proposition \ref{prop5} yields an
explicit formula:

\begin{corollary}
\label{cor2}Let $(S_{t})_{t}$ be a symplectic isotopy in $\operatorname*{Sp}%
(n)$. There exists a quadratic Hamiltonian function $H=H(z,t)$ such that
$S_{t}$ is the phase flow determined by the Hamilton equations $\dot
{z}=J\partial_{z}H$. The Hamiltonian function is the quadratic form%
\begin{equation}
H=-\frac{1}{2}J\dot{S}_{t}S_{t}^{-1}z\cdot z \label{hamzo}%
\end{equation}
where $\dot{S}_{t}=dS_{t}/dt$. In particular, if $S_{t}=e^{tX}$ with
$X\in\mathfrak{sp}(n)$ (the symplectic Lie algebra) then $H=-\frac{1}%
{2}JXz\cdot z.$
\end{corollary}

\begin{proof}
We have $\dot{f}_{t}\circ f_{t}^{-1}=\dot{S}_{t}S_{t}^{-1}$; applying formula
(\ref{hzt}) we get%
\begin{equation}
H(z,t)=-\int_{0}^{1}\sigma\left(  \dot{S}_{t}S_{t}^{-1}(\lambda z),z\right)
d\lambda\label{hamzobis}%
\end{equation}
which is (\ref{hamzo}), taking into account the linearity of $\sigma$ and
$S_{t}$.
\end{proof}

Writing $S_{t}$ and its inverse in block matrix form%
\begin{equation}
S_{t}=%
\begin{pmatrix}
A_{t} & B_{t}\\
C_{t} & D_{t}%
\end{pmatrix}
\text{ \ , \ }S_{t}^{-1}=%
\begin{pmatrix}
D_{t}^{T} & -B_{t}^{T}\\
-C_{t}^{T} & A_{t}^{T}%
\end{pmatrix}
\end{equation}
(the second formula following from the identity $S_{t}JS_{t}^{T}=J$) we have%
\[
J\dot{S}_{t}S_{t}^{-1}=%
\begin{pmatrix}
\dot{C}_{t}D_{t}^{T}-\dot{D}_{t}C_{t}^{T} & \dot{D}_{t}A_{t}^{T}-\dot{C}%
_{t}B_{t}^{T}\\
\dot{B}_{t}C_{t}^{T}-\dot{A}_{t}D_{t}^{T} & \dot{A}_{t}B_{t}^{T}-\dot{B}%
_{t}A_{t}^{T}%
\end{pmatrix}
;
\]
it follows from (\ref{hamzo}) that we have the explicit expression%
\begin{equation}
H=\tfrac{1}{2}(\dot{D}_{t}C_{t}^{T}-\dot{C}_{t}D_{t}^{T})x^{2}+(\dot{C}%
_{t}B_{t}^{T}-\dot{D}_{t}A_{t}^{T})px+\tfrac{1}{2}(\dot{B}_{t}A_{t}^{T}%
-\dot{A}_{t}B_{t}^{T})p^{2} \label{hamzi}%
\end{equation}
for the Hamiltonian function.

\begin{example}
\label{exam1}Consider the one-parameter family of matrices%
\[
S_{t}=%
\begin{pmatrix}
\cos\omega(t) & \sin\omega(t)\\
-\sin\omega(t) & \cos\omega(t)
\end{pmatrix}
\]
where $\omega$ is a $C^{2}$ function of time such that $\psi(0)=0$. We have
\[
J\dot{S}_{t}S_{t}^{-1}=%
\begin{pmatrix}
-\dot{\omega}(t) & 0\\
0 & -\dot{\omega}(t)
\end{pmatrix}
\]
hence $(S_{t})_{t}$ is the flow of the time-dependent harmonic oscillator
Hamiltonian%
\[
H(z,t)=\frac{\dot{\omega}(t)}{2}(p^{2}+x^{2}).
\]

\end{example}

The case of affine symplectic isotopies is a straightforward extension of the
result above:

\begin{corollary}
\label{cor3}Let $(S_{t})_{t}$ be a symplectic isotopy in $\operatorname*{Sp}%
(n)$ and $t\longmapsto z_{t}$ a $C^{1}$ path in $\mathbb{R}^{2n}$ with
$z_{0}=0$. The symplectic isotopy $(f_{t})_{t}$ defined by $f_{t}=S_{t}%
T(z_{t})$ where $T(z_{t})$ is the translation $z\longmapsto z+z_{t}$ is the
Hamiltonian flow determined by%
\begin{equation}
H(z,t)=-\frac{1}{2}J\dot{S}_{t}S_{t}^{-1}z\cdot z+\sigma\left(  z,S_{t}\dot
{z}_{t}\right)  . \label{hamzoter}%
\end{equation}

\end{corollary}

\begin{proof}
We have $f_{t}(z)=S_{z}(z+z_{t})$ hence%
\[
\dot{f}_{t}(z)=\dot{S}_{t}(z+z_{t})+S_{t}\dot{z}_{t}=T(S_{t}\dot{z}_{t}%
)\dot{S}_{t}T(z_{t})z
\]
hence $\dot{f}_{t}=T(S_{t}\dot{z}_{t})\dot{S}_{t}T(z_{t})$. The inverse
of\ $f_{t}$ being given by%
\[
f_{t}^{-1}=(S_{t}T(z_{t}))^{-1}=T(-z_{t})S_{t}^{-1}%
\]
we thus have $\dot{f}_{t}f_{t}^{-1}=T(S_{t}\dot{z}_{t})\dot{S}_{t}S_{t}^{-1}$.
Hence, by formula (\ref{hzt}),
\begin{align*}
H(z,t)  &  =-\int_{0}^{1}\sigma\left(  T(S_{t}\dot{z}_{t})\dot{S}_{t}%
S_{t}^{-1}(\lambda z),z\right)  d\lambda\\
&  =-\int_{0}^{1}\sigma\left(  \dot{S}_{t}S_{t}^{-1}(\lambda z),z\right)
d\lambda-\int_{0}^{1}\sigma\left(  S_{t}\dot{z}_{t},z\right)  d\lambda
\end{align*}
which is (\ref{hamzoter}) in view of formulas (\ref{hamzo}) and
(\ref{hamzobis}). and taking into account the antisymmetry of the symplectic form.
\end{proof}

The argument above can be easily reversed, yielding an explicit solution of
the Hamilton equations for a Hamiltonian function of the type
\begin{equation}
H(z,t)=\frac{1}{2}M(t)z^{2}+m(t)z. \label{hzm}%
\end{equation}
Let in fact $(S_{t})_{t}$ be the flow determined by the homogeneous part
$H_{0}(z,t)=\frac{1}{2}M(t)z^{2}$; we have (Hamilton's equations) $\dot{S}%
_{t}=JM(t)S_{t}$ hence $H_{0}(z,t)=-\frac{1}{2}J\dot{S}_{t}S_{t}^{-1}z$.
Setting $\dot{z}_{t}=S_{t}^{-1}Jm(t)$, we can thus rewrite $H$ in the form
(\ref{hamzoter}). Summarizing, the flow $(f_{t}^{H})_{t}$ determined by
(\ref{hzm}) is given by%
\begin{equation}
f_{t}^{H}=S_{t}T(z_{t})\text{ \ , \ }z_{t}=\int_{0}^{t}S_{t^{\prime}}%
^{-1}Jm(t^{\prime})dt^{\prime}. \label{fthm}%
\end{equation}
Let us illustrate this by a classical example: the time-dependent driven
harmonic oscillator.

\begin{example}
\label{exam2}Consider the Hamiltonian function
\[
H(z,t)=\frac{\omega(t)}{2}(p^{2}+x^{2})+f(t)x.
\]
In view of Example \ref{exam1} the flow $(S_{t})_{t}$ of the homogeneous part
is given by%
\begin{gather*}
S_{t}=%
\begin{pmatrix}
\cos\Omega(t) & \sin\Omega(t)\\
-\sin\Omega(t) & \cos\Omega(t)
\end{pmatrix}
\\
\Omega(t)=\int_{0}^{t}\omega(t^{\prime})dt^{\prime}+\Omega(0).
\end{gather*}
Setting $m(t)=(f(t),0)$.we have
\[
z_{t}=\left(
{\textstyle\int_{0}^{t}}
f(t^{\prime})\sin\Omega(t^{\prime})dt^{\prime},-%
{\textstyle\int_{0}^{t}}
f(t^{\prime})\cos\Omega(t)dt^{\prime}\right)
\]
and it suffices to apply formula (\ref{fthm}).
\end{example}

\subsubsection{The flow determined by $H_{0}+H_{1}$}

Assume we are given a \textquotedblleft master Hamiltonian\textquotedblright%
\ $H_{0}$, which we perturb by another Hamiltonian $H_{1}$ (an archetypical
example in the case $n=1$ would be the choice $H_{0}=p^{2}/2$ and
$H_{1}=V(x,t)$). The following result shows how to calculate the flow
$(f_{t}^{H})_{t}$ of $H=H_{0}+H_{1}$ knowing $(f_{t}^{H_{0}})$ and
$(f_{t}^{H_{1}})$.

\begin{proposition}
\label{prop6}Let $H=H_{0}+H_{1}$. We have
\begin{equation}
f_{t}^{H}=f_{t}^{H_{0}}f_{t}^{H_{1}^{t}}\text{ \ with }H_{1}^{t}%
(z,t)=H_{1}(f_{t}^{H_{0}}(z),t). \label{f}%
\end{equation}

\end{proposition}

\begin{proof}
Since $f_{0}^{H}$ and $f_{0}^{H_{0}}f_{0}^{H_{1}^{t}}$ are both the identity
on $\mathbb{R}^{2n}$ it suffices to show that the $t$-derivatives of
$f_{t}^{H}$ and $f_{t}^{H_{0}}f_{t}^{H_{1}^{t}}$ are equal. We have, using the
product and chain rules,
\begin{align*}
\frac{d}{dt}(f_{t}^{H_{0}}(f_{t}^{H_{1}^{t}}(z)))  &  =\frac{d}{dt}%
f_{t}^{H_{0}}(f_{t}^{H_{1}^{t}}(z))+Df_{t}^{H_{0}}(f_{t}^{H_{1}^{t}}%
(z))\frac{d}{dt}f_{t}^{H_{1}^{t}}(z)\\
&  =X_{H_{0}}(f_{t}^{H_{0}}f_{t}^{H_{1}^{t}}(z))+Df_{t}^{H_{0}}(f_{t}%
^{H_{1}^{t}}(z))X_{H_{1}^{t}}(f_{t}^{H_{1}^{t}}(z)).
\end{align*}
By definition, $H_{1}^{t}=H_{1}\circ f_{t}^{H_{0}}$ hence, using (\ref{XHg}),%
\begin{align*}
Df_{t}^{H_{0}}(f_{t}^{H_{1}^{t}}(z))X_{H_{1}^{t}}(f_{t}^{H_{1}^{t}}(z))  &
=Df_{t}^{H_{0}}(f_{t}^{H_{1}^{t}}(z))X_{H_{1}^{t}}((f_{t}^{H_{0}})^{-1}%
f_{t}^{H_{0}}f_{t}^{H_{1}^{t}}(z))\\
&  =X_{H_{1}^{t}\circ(f_{t}^{H_{0}})^{-1}}(f_{t}^{H_{0}}f_{t}^{H_{1}^{t}%
}(z))\\
&  =X_{H_{1}}(f_{t}^{H_{0}}f_{t}^{H_{1}^{t}}(z)).
\end{align*}
Summarizing,%
\begin{align*}
\frac{d}{dt}(f_{t}^{H_{0}}(f_{t}^{H_{1}^{t}}(z)))  &  =X_{H_{0}}(f_{t}^{H_{0}%
}f_{t}^{H_{1}^{t}}(z))+X_{H_{1}}(f_{t}^{H_{0}}f_{t}^{H_{1}^{t}}(z))\\
&  =X_{H_{0}+H_{1}}(f_{t}^{H_{0}}(f_{t}^{H_{1}^{t}}(z)))
\end{align*}
which we set out to prove.
\end{proof}

Assume, in particular, that $H$ is separable, that is $H=H_{0}+H_{1}$\ with
$H_{0}(z)=T(p)$ and $H_{1}(z)=V(x)$. The flow $(f_{t}^{H_{0}})$ is given by%
\begin{equation}
f_{t}^{H_{0}}(x,p)=(x+t\partial_{p}T(p),p) \label{sep1}%
\end{equation}
and we thus have%
\begin{equation}
H_{1}^{t}(x,p)=V(x+t\partial_{p}T(p)) \label{h1t}%
\end{equation}
so that $(f_{t}^{H_{1}})$ is obtained by solving the Hamilton equations%
\begin{align}
\dot{x}  &  =t\partial_{x}V(x+t\partial_{p}T(p))D_{p}^{2}T(p)\label{x1t}\\
\dot{p}  &  =-\partial_{x}V(x+t\partial_{p}T(p)) \label{p1t}%
\end{align}
where $D_{p}^{2}T(p)$ is the Hessian matrix (= matrix of second derivatives)
of $T$ calculated at $p$. When $T(p)=\frac{1}{2}p^{2}$ is the ordinary kinetic
energy these formulas reduce to the simple system%
\begin{equation}
\dot{x}=t\partial_{x}V(x+tp)\text{ \ , \ }\dot{p}=-\partial_{x}V(x+tp).
\end{equation}
Setting $u=x+tp$ this system is equivalent to the equations $\ddot{u}%
+\partial_{u}V(u)=0$ and\ $p=\dot{u}$, that is to Hamilton's equations for
$H$. There are potential applications of this result to the theory of
symplectic integrators (see\textit{ e.g.} Chorin \textit{et al}.
\cite{chorin}, McLachlan and Atela \cite{mac}, Wang \cite{Wang}).

Let us illustrate these formulas on an elementary example.

\begin{example}
\label{exa1}Assume $n=1$ and let $H_{0}(x,p)=\frac{1}{2}p^{2}$ and
$H_{1}(x,p)=\frac{1}{2}x^{2}$. The Hamiltonian $(f_{t}^{H_{0}})$ flow
determined by $H_{0}$ is identified with the one-parameter group of symplectic
matrices $S_{t}=%
\begin{pmatrix}
1 & t\\
0 & 1
\end{pmatrix}
$. We thus have $H_{1}^{t}(x,p)=\frac{1}{2}(x+pt)^{2}$ and the Hamilton
equations are $\dot{x}=(x+pt)t$ and $\dot{p}=-(x+pt)$ hence the flow
determined by $H_{1}^{t}$ consists of the symplectic matrices%
\begin{equation}
f_{t}^{H_{1}^{t}}=%
\begin{pmatrix}
\cos t+t\sin t & \sin t-t\cos t\\
-\sin t & \cos t
\end{pmatrix}
\label{mat1}%
\end{equation}
and we thus have
\begin{equation}
f_{t}^{H_{0}}f_{t}^{H_{1}^{t}}=%
\begin{pmatrix}
\cos t & \sin t\\
-\sin t & \cos t
\end{pmatrix}
\label{mat2}%
\end{equation}
which is the flow of the harmonic oscillator Hamiltonian $H(z)=\frac{1}%
{2}(p^{2}+x^{2})$.
\end{example}

\subsection{The group $\operatorname*{Ham}(n)$ and its universal covering}

\subsubsection{Products of Hamiltonian isotopies and $\operatorname*{Ham}(n)$}

Let $(f_{t}^{H})_{t}$ and $(f_{t}^{K})_{t}$ be Hamiltonian flows; we assume
that they exist for $t\in I_{T}=[-T,T]$. It follows from Proposition
\ref{prop6} above that%
\begin{align}
f_{t}^{H}f_{t}^{K}  &  =f_{t}^{H\#K}\text{ \ \ \textit{with} \ \ }%
H\#K(z,t)=H(z,t)+K((f_{t}^{H})^{-1}(z),t).\label{ch1}\\
(f_{t}^{H})^{-1}  &  =f_{t}^{\bar{H}}\text{ \ \ \textit{with} \ \ }\bar
{H}(z,t)=-H(f_{t}^{H}(z),t). \label{ch2}%
\end{align}
In fact, setting $H_{0}(z,t)=H(z,t)$ and $H_{1}=K((f_{t}^{H})^{-1}(z),t)$ we
have $H_{0}+H_{1}=H\#K$ hence, applying formula (\ref{f}),%
\[
f_{t}^{H\#K}=f_{t}^{H_{0}+H_{1}}=f_{t}^{H_{0}}f_{t}^{H_{1}^{t}}=f_{t}^{H}%
f_{t}^{K}%
\]
which is (\ref{ch1}). Formula (\ref{ch2}) immediately follows noting that
$f_{t}^{H}f_{t}^{\bar{H}}=I_{\mathrm{d}}$ (for an alternative proof see Hofer
and Zehnder \cite{HZ}, p.144).

We next assume that all Hamiltonians $H$ are constant outside a compact subset
$\mathcal{K}$ of $\mathbb{R}^{2n}$. We will call such Hamiltonian functions
\textit{compactly supported} (this is a slight, but innocuous, abuse of
terminology). The vector field $X_{H}$ determined by such a Hamiltonian
function $H$ is complete (see for instance Abraham and Marsden \cite{AM}) and
hence the flow $(f_{t}^{H})_{t}$ exists for all times $t$. Hamiltonian flows
associated with such Hamiltonian functions are the identity outside the
compact set $\mathcal{K}$. Hamiltonian symplectomorphisms coming from
compactly supported Hamiltonians form a subgroup $\operatorname*{Ham}(n)$ of
the group $\operatorname*{Symp}(n)$ of all symplectomorphisms; it is in fact a
normal subgroup of $\operatorname*{Symp}(n)$ as follows from the conjugation
formula (\ref{FHg}). That $\operatorname*{Ham}(n)$ is a group follows from the
two formulas (\ref{ch1}) and (\ref{ch2}) above. It turns out that
$\operatorname*{Ham}(n)$ is a connected group. To prove this it suffices to
show that every $f\in\operatorname*{Ham}(n)$ can be joined to the identity
$I_{\mathrm{d}}$ by a path in $\operatorname*{Ham}(n)$. But by definition of
$\operatorname*{Ham}(n)$ there exists a Hamiltonian function $H$ such that
$f=f_{1}^{H}$; the path we are looking for is just $(f_{t}^{H})_{0\leq t\leq
1}$. It follows, using Proposition \ref{prop5}, that $\operatorname*{Ham}(n)$
is the connected component of the group $\operatorname*{Symp}_{\mathrm{c}}(n)$
of all compactly supported canonical transformations, \emph{i.e}. equal to the
identity outside a compact subset of $\mathbb{R}^{2n}$ (for detail see Hofer
and Zehnder \cite{HZ})..

\subsubsection{The universal covering and a homotopy result}

Consider the universal covering group $\widetilde{\operatorname*{Ham}}(n)$
\cite{HZ} of $\operatorname*{Ham}(n)$; the elements of
$\widetilde{\operatorname*{Ham}}(n)$ consist of homotopy classes (with fixed
endpoints) of Hamiltonian isotopies $(f_{t}^{H})_{t\in I_{T}}$ and the group
law is given by%
\begin{equation}
\lbrack f_{t}^{H}][f_{t}^{K}]=[f_{t}^{H\#K}] \label{group1}%
\end{equation}
where $[f_{t}^{H}]$ is the homotopy class (with fixed endpoints) of the
isotopy $(f_{t}^{H})_{t\in I_{T}}$. The following result shows that the
product in $\widetilde{\operatorname*{Ham}}(n)$ can be defined using either
the relation (\ref{group1}) , or by concatenation of Hamiltonian isotopies:

\begin{proposition}
Let $(f_{t}^{H})_{t}$ and $(f_{t}^{K})_{t}$ be two Hamiltonian isotopies. (i)
The paths $(f_{t}^{H\#K})_{0\leq t\leq2t_{0}}$ and $(f_{t}^{H\Diamond
K})_{1\leq t\leq2t_{0}}$ where%
\begin{equation}
f_{t}^{H\Diamond K}=\left\{
\begin{array}
[c]{c}%
f_{t}^{K}\text{ \ for \ }0\leq t\leq t_{0}\\
f_{t-t_{0}}^{H}f_{t_{0}}^{K}\text{ \ for \ }t_{0}\leq t\leq2t_{0}%
\end{array}
\right.  \label{fthomotopy}%
\end{equation}
are homotopic with fixed endpoints $I_{\mathrm{d}}$ and $f_{t_{0}}%
^{H\#K}=f_{t_{0}}^{H\Diamond K}$. (ii) The group law of
$\widetilde{\operatorname*{Ham}}(n)$ can thus be defined by%
\begin{equation}
\lbrack f_{t}^{H}][f_{t}^{K}]=[f_{t}^{H\Diamond K}]. \label{group2}%
\end{equation}

\end{proposition}

\begin{proof}
(i) Rescaling time if necessary, it suffices to consider the case $t_{0}%
=\frac{1}{2}$. Let us construct explicitly a homotopy taking the path
$(f_{t}^{H}f_{t}^{K})_{0\leq t\leq1}$ to the path $(f_{t}^{H\Diamond
K})_{0\leq t\leq1}$, that is, a continuous mapping
\[
h:[0,1]\times\lbrack0,1]\longrightarrow\operatorname*{Ham}(n)
\]
such that $h(t,0)=f_{t}^{H}f_{t}^{K}$ and\ $h(t,1)=f_{t}^{H\Diamond K}$ for
$0\leq t\leq1$. Since we want to preserve endpoints during the deformation, we
require in addition that $h(0,s)=I_{\mathrm{d}}$ and $h(1,s)=f_{1}^{H}$ for
all $s\in\lbrack0,1]$. Define $h$ by $h(t,s)=a(t,s)b(t,s)$ where $a$ and $b$
are the functions%
\[
a(t,s)=\left\{
\begin{array}
[c]{c}%
I_{\mathrm{d}}\text{ \ for \ }0\leq t\leq\frac{s}{2}\\
f_{(2t-s)/(2-s)\text{ \ \ }}^{H}\text{\ for \ }\frac{s}{2}\leq t\leq1
\end{array}
\right.
\]
and%
\[
b(t,s)=\left\{
\begin{array}
[c]{c}%
f_{2t/(2-s)}^{K}\text{ \ }\ \text{for \ }0\leq t\leq1-\frac{s}{2}\\
f_{1\text{\ \ }}^{K}\text{ \ for \ }\frac{s}{2}\leq t\leq1.
\end{array}
\right.
\]
We have $a(t,0)=f_{t}^{H}$, $b(t,0)=f_{t}^{K}$ hence $h(t,0)=f_{t}^{H}%
f_{t}^{K}$; similarly
\[
h(t,1)=\left\{
\begin{array}
[c]{c}%
f_{2t}^{K}\text{ \ }\ \text{for \ }0\leq t\leq\frac{1}{2}\\
f_{2t-1}^{H}f_{1}^{K}\text{ \ for \ }\frac{1}{2}\leq t\leq1
\end{array}
\right.
\]
that is $h(t,1)=f_{t}^{H\Diamond K}$. The relations $h(0,s)=I_{\mathrm{d}}$
and $h(1,s)=f_{1}^{H}$ for $0\leq s\leq1$ are easily verified. ii) That the
group law of $\widetilde{\operatorname*{Ham}}(n)$ can be defined by
(\ref{group2}) follows from (i) and (\ref{group1}).
\end{proof}

One checks that the path $(f_{t}^{H\Diamond K})_{0\leq t\leq t_{0}}$ is an
isotopy corresponding to the Hamiltonian
\begin{equation}
H\Diamond K(z,t)=\left\{
\begin{array}
[c]{c}%
K(z,t)\text{ \ for \ }0\leq t\leq t_{0}\\
H(z,t-t_{0})\text{\ for \ }t_{0}\leq t\leq2t_{0}%
\end{array}
\right.  \label{hkdiam}%
\end{equation}
which is discontinuous at $t=t_{0}$. In fact, using formula (\ref{hzt}) in
Proposition \ref{prop5} we have%
\[
H\Diamond K(z,t)=-\int_{0}^{1}\sigma\left(  \tfrac{d}{dt}f_{t}^{K}\circ
(f_{t}^{K})^{-1}(\lambda z),z\right)  d\lambda=K(z,t)
\]
for $0\leq t\leq t_{0}$, and similarly
\begin{align*}
H\Diamond K(z,t)  &  =-\int_{0}^{1}\sigma\left(  \tfrac{d}{dt}(f_{t-t_{0}}%
^{H}f_{t_{0}}^{K})\circ(f_{t-t_{0}}^{H}f_{t_{0}}^{K})^{-1}(\lambda
z),z\right)  d\lambda\text{ }\\
&  =-\int_{0}^{1}\sigma\left(  \tfrac{d}{dt}f_{t-t_{0}}^{H}\circ(f_{t-t_{0}%
}^{H})^{-1}(\lambda z),z\right)  d\lambda\\
&  =H(z,t-t_{0})
\end{align*}
for $t_{0}\leq t\leq2t_{0}$.

\section{The Groups $\operatorname*{Mp}(n)$ and $\operatorname*{IMp}%
(n)$\label{sec2}}

The symplectic group $\operatorname*{Sp}(n)$ has covering groups
$\operatorname*{Sp}_{q}(n)$ of all orders $q=2,...,\infty$. Among all these
the double covering $\operatorname*{Sp}_{2}(n)$ plays a very special role in
mathematics and physics, because it can be faithfully represented by a group
of unitary operators acting on $L^{2}(\mathbb{R}^{n})$. This group is called
the metaplectic group and denoted by $\operatorname*{Mp}(n)$. We will loosely
speak, as is usual in the literature, about the \textquotedblleft metaplectic
representation $\mu$ of the symplectic group\textquotedblright; one should
keep in mind that, strictly speaking, $\mu$ is a representation
$\operatorname*{Sp}_{2}(n)\longrightarrow\operatorname*{Mp}(n)$ of the double
cover of $\operatorname*{Sp}(n)$.

\subsection{Definition and main properties}

The symplectic group $\operatorname*{Sp}(n)$ is a connected Lie group
contractible to the unitary group $U(n)$, hence $\pi_{1}[\operatorname*{Sp}%
(n)]\simeq(\mathbb{Z},+)$. It follows that $\operatorname*{Sp}(n)$ has
coverings $\Pi_{q}:\operatorname*{Sp}_{q}(n)\longrightarrow\operatorname*{Sp}%
(n)$ of all orders $q=2,...,\infty$. An essential fact is that the double
covering $\operatorname*{Sp}_{2}(n)$ has a faithful representation as a group
of unitary operators acting on $L^{2}(\mathbb{R}^{n})$; this group is the
metaplectic group $\operatorname*{Mp}(n)$.

\subsubsection{Quadratic Fourier transforms and generating functions}

The group $\operatorname*{Mp}(n)$ is generated by the quadratic Fourier
transforms $\widehat{S}_{W,m}$ defined as follows: let
\begin{equation}
W(x,x^{\prime})=\tfrac{1}{2}Px^{2}-Lx\cdot x^{\prime}+\tfrac{1}{2}Qx^{\prime2}
\label{W}%
\end{equation}
be a real quadratic form where $P=P^{T}$, $Q=Q^{T}$, $\det L\neq0$ and set
$\Delta_{m}(W)=i^{m}\sqrt{|\det L|}$ where the integer $m$ corresponds to a
choice of $\arg\det L$. For $\psi\in\mathcal{S}(\mathbb{R}^{n})$%
\begin{equation}
\widehat{S}_{W,m}\psi(x)=\left(  \tfrac{1}{2\pi i\hbar}\right)  ^{n/2}%
\Delta_{m}(W)\int_{\mathbb{R}^{n}}e^{\frac{i}{\hbar}W(x,x^{\prime})}%
\psi(x^{\prime})dx^{\prime} \label{SWm}%
\end{equation}
(with the convention $\arg i=\pi/2$). In fact:

\begin{proposition}
\label{prop1}(i) Every $\widehat{S}\in\operatorname*{Mp}(n)$ can be written
(non-uniquely) as the product $\widehat{S}_{W,m}\widehat{S}_{W^{\prime
},m^{\prime}}$ of two quadratic Fourier transforms. (ii) We have
\[
\widehat{S}_{W,m}\widehat{S}_{W^{\prime},m^{\prime}}=\widehat{S}%
_{W^{\prime\prime},m^{\prime\prime}}%
\]
if and only if $\det(P^{\prime}+Q)\neq0$ and in this case%
\begin{align}
P^{\prime\prime}  &  =P-L^{T}(P^{\prime}+Q)^{-1}L\label{P''}\\
L^{\prime\prime}  &  =L^{\prime}(P^{\prime}+Q)^{-1}L\label{L''}\\
Q^{\prime\prime}  &  =Q^{\prime}-L^{\prime}(P^{\prime}+Q)^{-1}(L^{\prime}%
)^{T}; \label{Q''}%
\end{align}
and the Maslov index of $\widehat{S}_{W^{\prime\prime},m^{\prime\prime}}$ is
given by
\begin{equation}
m^{\prime\prime}\equiv m^{\prime}+m^{\prime}-\operatorname*{Inert}(P^{\prime
}+Q)\text{ \ }\operatorname{mod}4. \label{mmm}%
\end{equation}
(iii) We have $(\widehat{S}_{W,m})^{-1}=\widehat{S}_{W^{\prime},m^{\prime}}$
where $W^{\prime}(x,x^{\prime})=-W(x^{\prime},x)$ and $m^{\prime}\equiv n-m$,
$\operatorname{mod}4$.
\end{proposition}

\begin{proof}
For (i) and formula (\ref{mmm}) in (iii), see Leray \cite{Leray}, de Gosson
\cite{AIF,Birk}; formulas (\ref{P''})--(\ref{Q''}) are obtained by matrix
multiplication using (\ref{swplq}).
\end{proof}

\begin{remark}
Formula (\ref{mmm}) identifies the integer $m$ with the Maslov index modulo 4
on $\operatorname*{Mp}(n)$; see Leray \cite{Leray}, Souriau \cite{Souriau}, de
Gosson \cite{AIF,JMPA,Birk}.
\end{remark}

The covering projection $\Pi:\operatorname*{Mp}(n)\longrightarrow
\operatorname*{Sp}(n)$ is defined by its action on the generators
$\widehat{S}_{W,m}$:

\begin{proposition}
\label{prop2}We have $\Pi(\widehat{S}_{W,m})=S_{W}$ where $S_{W}%
\in\operatorname*{Sp}(n)$ is generated by the quadratic form $W$, that is
\begin{equation}
(x,p)=S_{W}(x^{\prime},p^{\prime})\Longleftrightarrow\left\{
\begin{array}
[c]{c}%
p=\partial_{x}W(x,x^{\prime})\\
p^{\prime}=-\partial_{x^{\prime}}W(x,x^{\prime})
\end{array}
\right.  . \label{swpp}%
\end{equation}

\end{proposition}

When $W$ is given by (\ref{W}) we have the explicit formula%
\begin{equation}
S_{W}=%
\begin{pmatrix}
L^{-1}Q & L^{-1}\\
PL^{-1}Q-L^{T} & PL^{-1}%
\end{pmatrix}
\text{.} \label{swplq}%
\end{equation}
Equivalently, if
\begin{equation}
S_{W}=%
\begin{pmatrix}
A & B\\
C & D
\end{pmatrix}
\text{ \ , \ }\det B\neq0 \label{SW}%
\end{equation}
is a free symplectic matrix, then its generating function is%
\begin{equation}
W(x,x^{\prime})=\tfrac{1}{2}DB^{-1}x^{2}-B^{-1}x\cdot x^{\prime}+\tfrac{1}%
{2}B^{-1}Ax^{\prime2}. \label{WABCD}%
\end{equation}

For $z_{0}=(x_{0},p_{0})\in\mathbb{R}^{2n}$ the Heisenberg--Weyl operator
$\widehat{T}(z_{0})$ is the time-one propagator of Schr\"{o}dinger's equation
associated with the translation Hamiltonian $z\longmapsto\sigma(z,z_{0})$,
that is, formally, $\widehat{T}(z_{0})=e^{-i\sigma(\widehat{z},z_{0})/\hbar}$
where%
\begin{equation}
\sigma(\widehat{z},z_{0})=\widehat{p}\cdot x_{0}-p_{0}\cdot\widehat{x}
\label{sighat}%
\end{equation}
where $\widehat{p}=-i\hbar\partial_{x}$ and $\widehat{x}$ is multiplication by
$x$. It is hence a unitary operator on $L^{2}(\mathbb{R}^{n})$ whose action is
explicitly given by
\begin{equation}
\widehat{T}(z_{0})\psi(x)=e^{\frac{i}{\hbar}(p_{0}\cdot x-\frac{1}{2}%
p_{0}\cdot x_{0})}\psi(x-x_{0}). \label{hw1}%
\end{equation}
The Heisenberg--Weyl operators satisfy the product relations%
\begin{align}
\widehat{T}(z_{0})\widehat{T}(z_{1})  &  =e^{\frac{i}{\hbar}\sigma(z_{0}%
,z_{1})}\widehat{T}(z_{1})\widehat{T}(z_{0})\label{hw2}\\
\widehat{T}(z_{0}+z_{1})  &  =e^{-\frac{i}{2\hbar}\sigma(z_{0},z_{1}%
)}\widehat{T}(z_{0})\widehat{T}(z_{1}) \label{hw3}%
\end{align}
(the first formula follows from the second, interchanging $z_{0}$ and $z_{1}%
$). Let $\widehat{S}\in\operatorname*{Mp}(n)$ and $S=\Pi(\widehat{S})$; we
have the important intertwining formula%
\begin{equation}
\widehat{S}\widehat{T}(z_{0})\widehat{S}^{-1}=\widehat{T}(Sz_{0})
\label{sympco1}%
\end{equation}
which can be proven \cite{Birk,Birkbis} by checking it on the quadratic
Fourier transforms $\widehat{S}_{W,m}$ generating $\operatorname*{Mp}(n)$. We
note that this formula is (incorrectly) taken in some texts as a definition of
metaplectic operators; the irreducibility of the Schr\"{o}dinger
representation and Stone--von Neumann's theorem are invoked to motivate this
\textquotedblleft definition\textquotedblright. However, (\ref{sympco1}) only
defines a \textit{projective} representation of $\operatorname*{Sp}(n)$, and
not its double covering group; if one wants a true covering group one has to
carefully determine the cocycle associated with this projective representation
(see Reiter \cite{Reiter} for a detailed analysis). In \cite{AIF} we have
given an analytical construction of this cocycle, and in
\cite{go09,gogo03,gogo06} we use a topological and cohomological method, using
techniques from symplectic geometry (the Leray index, and the notion of
signature of a triple of Lagrangian planes, further developed in
\cite{Birk,Birkbis}).

\subsubsection{The inhomogeneous metaplectic group}

Consider now the Heisenberg group $\mathbb{H}(n)$: it is $\mathbb{R}%
^{2n}\times\mathbb{R}$ equipped with the group law%
\[
(z,t)(z^{\prime},t^{\prime})=(z+z^{\prime},t+t^{\prime}+\tfrac{1}{2}%
\sigma(z,z^{\prime}));
\]
the mapping $\rho:\mathbb{H}(n)\longrightarrow\mathcal{U}(L^{2}(\mathbb{R}%
^{n}))$ (the unitary operators on $L^{2}(\mathbb{R}^{n})$) defined by
\begin{equation}
\rho(z_{0},t_{0})\psi=e^{\frac{i}{\hbar}t_{0}}\widehat{T}(z_{0})\psi
\label{schrep1}%
\end{equation}
is a unitary and irreducible representation of $\mathbb{H}_{n}$ called the
\textit{Schr\"{o}dinger representation} of $\mathbb{H}(n)$. Defining the
action of the symplectic group $\operatorname*{Sp}(n)$ on $\mathbb{H}(n)$ by
$S(z,t)=(Sz,t)$ the group $\operatorname{WSp}(n)$ is the semi-direct product
of $\operatorname*{Sp}(n)$ and $\mathbb{H}_{n}$:%
\[
\operatorname{WSp}(n)=\operatorname*{Sp}(n)\ltimes\mathbb{H}(n)
\]
with group law given by%
\[
(S,u)(S^{\prime},u^{\prime})=(SS^{\prime},u(Su^{\prime}))
\]
for $u=(z,t)$ and $u^{\prime}=(z^{\prime},t^{\prime})$ (we mention that Burdet
\textit{et al}. \cite{bu} give a detailed study of $\operatorname{WSp}(n)$ and
its generating functions).

Let $\mu:\operatorname*{Sp}_{2}(n)\longrightarrow\operatorname*{Mp}(n)$ be the
metaplectic representation of the double cover of the symplectic group and let
us define
\[
\chi:\operatorname*{Sp}\nolimits_{2}(n)\ltimes\mathbb{H}(n)\longrightarrow
\mathcal{U}(L^{2}(\mathbb{R}^{n}))
\]
by $\chi=\mu\otimes\rho$:%
\[
\chi(\widetilde{S},u)=\mu(\widetilde{S})\rho(u)\text{ \ , \ }\widetilde{S}%
\in\operatorname*{Sp}\nolimits_{2}(n)\text{ \ , \ }u=(z,t).
\]
We have:

\begin{proposition}
$\chi$ is the unique unitary representation of the universal cover of
$\operatorname{WSp}(n)$ whose restriction to the Heisenberg group
$\mathbb{H}(n)$ is $\rho$.
\end{proposition}

\noindent We refer to Folland \cite{fo89} (p. 196) for a detailed elementary
proof of this result. The group defined by this extended representation is
generated by the metaplectic operators and Heisenberg--Weyl operators, and is
called the \emph{inhomogeneous metaplectic group} $\operatorname*{IMp}(n)$. It
follows from the intertwining relation (\ref{sympco1}) that every element of
$\operatorname*{IMp}(n)$ can be written $\widehat{T}(z_{0})\widehat{S}$ (resp.
$\widehat{S}\widehat{T}(z_{0})$) for some $\widehat{S}\in\operatorname*{Mp}%
(n)$ and $z_{0}\in\mathbb{R}^{2n}$. See Binz and Pods \cite{Binz} for a
detailed study of the Heisenberg and metaplectic representations with a
detailed study of applications to optics, quantization, and field quantization.

\subsection{The path lifting property for $\operatorname*{Mp}(n)$%
\label{seclift}}

Here is a fundamental property of the metaplectic representation, which is
often undeservedly ignored in the physical literature.

\subsubsection{Lifting paths to the covering group}

To understand this, let us first state the path lifting property for covering
spaces in its full generality (see \textit{e.g.} Spanier \cite{spanier}). Let
$\pi:Y\longrightarrow X$ be a covering projection $(Y$ and hence $X$ are
assumed to be connected) and $\gamma:[a,b]\longrightarrow X$ a continuous path
such \ that $\gamma(t_{0})=x_{0}=\pi(y_{0})$ for some $t_{0}\in\lbrack a,b]$.
Then there exists a unique continuous path $\widetilde{\gamma}%
:[a,b]\longrightarrow Y$ such that $\pi\circ\widetilde{\gamma}=\gamma$ and
$\widetilde{\gamma}(t_{0})=y_{0}$. Let us now choose $X=\operatorname*{Sp}%
(n)$, $Y=\operatorname*{Mp}(n)$ and assume that $0\in\lbrack a,b]$. We choose
for $\gamma$ a symplectic isotopy $(S_{t})_{t}$ in $\operatorname*{Sp}(n)$
(thus $S_{0}=I$). It follows from the path lifting property that there exists
a \emph{unique} path $\widehat{\gamma}$ of metaplectic operators
$\widehat{S}_{t}\in\operatorname*{Mp}(n)$ passing through the identity in
$\operatorname*{Mp}(n)$ at time $t_{0}=0$ and such that $\Pi(\widehat{S}%
_{t})=S_{t}$ for every $t\in\lbrack a,b]$. One shows \cite{Birk,Birkbis},
using the fact that the Lie algebras of $\operatorname*{Sp}(n)$ and
$\operatorname*{Mp}(n)$ are isomorphic, that the path $t\longmapsto
\widehat{\gamma}(t)=\widehat{S}_{t}$ is $C^{1}$ if $t\longmapsto
\gamma(t)=S_{t}$ is, and that it satisfies Schr\"{o}dinger's equation%
\begin{equation}
i\hbar\frac{d}{dt}\widehat{S}_{t}=\widehat{H}\widehat{S}_{t}\text{ \ ,
}\widehat{S}_{0}=I_{\mathrm{d}} \label{st}%
\end{equation}
where $\widehat{H}$ is the quantization of the (time-dependent) Hamiltonian
function $H$ determined from $(S_{t})_{t}$ using formula (\ref{hamzo}) in
Corollary \ref{cor2}. By \textquotedblleft quantization\textquotedblright, we
mean here the operator obtained by applying the Weyl correspondence to $H$
(this will be detailed in Section \ref{secweyl}). Summarizing:

\begin{proposition}
\label{prop9}Let $H$ be a Hamiltonian function of the type
\[
H(z,t)=\frac{1}{2}M(t)z^{2}%
\]
where $M(t)\in\operatorname*{Sym}(2n,\mathbb{R})$ depends in a $C^{j}$
($j\geq2$) fashion on $t\in\mathbb{R}$. The solution of the Schr\"{o}dinger
equation%
\[
i\hbar\frac{\partial\psi}{\partial t}=\widehat{H}\psi\text{ \ , \ }\psi
(\cdot,0)=\psi_{0}\in L^{2}(\mathbb{R}^{n})
\]
is given by $\psi(x,t)=\widehat{S}_{t}\psi_{0}(x)$ where $(\widehat{S}%
_{t})_{t}$ is the unique path of operators $\widehat{S}_{t}\in
\operatorname*{Mp}(n)$ such that $\Pi(\widehat{S}_{t})=S_{t}$ where
$(S_{t})_{t}$ is the flow determined by Hamilton's equations for $H$.
\end{proposition}

\begin{example}
\label{exa3}Let us illustrate this procedure on the harmonic oscillator
Hamiltonian $H(z)=\frac{1}{2}(p^{2}+x^{2})$. The flow $(f_{t}^{H})_{t}$
consists of the rotations
\[
f_{t}^{H}=%
\begin{pmatrix}
\cos t & \sin t\\
-\sin t & \cos t
\end{pmatrix}
.
\]
Using formulas (\ref{WABCD}) and (\ref{SWm}) the metaplectic operators with
projection $f_{t}^{H}$ are, for $t\notin\pi\mathbb{Z}$, the quadratic Fourier
operators%
\begin{equation}
\widehat{S}_{t}^{H}\psi(x)=\left(  \frac{1}{2\pi i\hbar}\right)  ^{1/2}%
\frac{i^{m}}{\sqrt{|\sin t|}}\int_{-\infty}^{\infty}\exp\left[  \frac{i}%
{\hbar}\frac{(x^{2}+x^{\prime2})\cos t-2xx^{\prime}}{2\sin t}\right]
\psi(x^{\prime})dx^{\prime}. \label{sth}%
\end{equation}
One shows (Souriau \cite{Souriau}, de Gosson \cite{AIF,Birk}) that if we
choose the value $m=-[t/\pi]$ ($[\cdot]$ the integer part function) for the
Maslov index, then $(\widehat{S}_{t}^{H})$ is the lift of $(f_{t}^{H})_{t}$ to
$\operatorname*{Mp}(n)$, and is hence the solution of Schr\"{o}dinger's
equation (\ref{st}) with Hamiltonian operator $\frac{1}{2}(-\hbar^{2}%
\partial_{x}^{2}+x^{2})$.
\end{example}

\subsubsection{Translation along a path}

Recall that the Heisenberg--Weyl operator $\widehat{T}(z_{0})$ is the time-one
evolution operator of Schr\"{o}dinger's equation%
\[
i\hbar\frac{\partial\psi}{\partial t}=\sigma(z_{0},\widehat{z})\psi.
\]
where $\widehat{z}=(x,-i\hbar\partial_{x})$. The following result (Littlejohn
\cite{Littlejohn}) is an extension to the case where $z_{0}$ depends
explicitly on $t$:

\begin{proposition}
\label{prop10}Let $t\longmapsto z_{0}(t)$, $z_{0}(0)=z_{0}$ be a $C^{1}$ path
in the phase space $\mathbb{R}^{2n}$. The solution of Schr\"{o}dinger's
equation%
\begin{equation}
i\hbar\frac{\partial\psi}{\partial t}=\sigma(\widehat{z},\dot{z}_{0}%
(t))\psi\text{ \ , \ }\psi(\cdot,0)=\psi_{0} \label{schrimp1}%
\end{equation}
is given by%
\begin{equation}
\psi(x,t)=e^{\frac{i}{\hbar}\chi(t)}\widehat{T}(z_{0}(t))\psi_{0}(x)
\label{schrimp2}%
\end{equation}
where the phase $\chi$ is real and given by
\begin{equation}
\chi(t)=-\frac{1}{2}\int_{0}^{t}\sigma(z_{0}(t^{\prime}),\dot{z}_{0}%
(t^{\prime}))dt^{\prime}. \label{schrimp3}%
\end{equation}

\end{proposition}

\begin{proof}
The idea is to write the evolution operator for the Schr\"{o}dinger equation
(\ref{schrimp1}) as the limit of a \textquotedblleft time-ordered
product\textquotedblright. Let $t_{0}=0,t_{1},...,t_{N}=t$ be a sequence of
successive times with $\Delta t=\sup_{j}|t_{j+1}-t_{j}|$ and set $z_{0}%
=z_{0}(0),z_{1}=z(t_{1}),...,z_{N}=z_{0}(t)$. Using the product formula
(\ref{hw3}) for the Heisenberg--Weyl operators we have
\begin{multline*}
\widehat{T}(z_{N}-z_{N-1})\cdot\cdot\cdot\widehat{T}(z_{2}-z_{1}%
)\widehat{T}(z_{1}-z_{0})\widehat{T}(z_{0})\\
=\exp\left(  \frac{i}{2\hbar}\sum_{k=0}^{N-1}\sigma(z_{k+1}-z_{k}%
,z_{k})\right)  \widehat{T}(z_{N})
\end{multline*}
and one finds that in the limit $N\rightarrow\infty$ (\textit{i.e}. $\Delta
t\rightarrow0$) this product converges to the operator
\[
U_{t}^{\sigma}=\exp\left(  -\frac{i}{2\hbar}\int_{0}^{t}\sigma_{0}%
(z(t^{\prime}),\dot{z}_{0}(t^{\prime}))dt^{\prime}\right)  \widehat{T}(z(t));
\]
let us check that $\psi=U_{t}^{\sigma}\psi_{0}$ indeed is the solution of
(\ref{schrimp1}). Writing (\ref{schrimp2}) as%
\[
\psi(x,t)=e^{\frac{i}{\hbar}\Phi(x,t)}\psi(x-x_{0}(t))
\]
the phase $\Phi$ being given by
\[
\Phi(x,t)=-\tfrac{1}{2}\int_{0}^{t}\sigma(z_{0}(t^{\prime}),\dot{z}%
_{0}(t^{\prime}))dt^{\prime}-\tfrac{1}{2}p_{0}(t)x_{0}(t)+p_{0}(t)x
\]
differentiation with respect to $t$ yields
\[
i\hbar\frac{\partial\psi}{\partial t}=\left[  \sigma(z_{0}(t^{\prime}),\dot
{z}_{0}(t^{\prime}))\psi(x-x(t))-i\hbar\dot{x}_{0}(t)\partial_{x}\psi
(x-x_{0}(t))\right]  e^{\frac{i}{\hbar}\Phi(x,t)}.
\]
A similar calculation, differentiating with respect to the $x$ variables shows
that the right hand side is precisely $\sigma(\widehat{z},\dot{z}%
(t))\psi(x,t)$.
\end{proof}

Suppose that $z(T)=z(0)=0$ and set $\gamma(t)=z(t)$ for $0\leq t\leq T$; then
\[
\chi(T)=-\frac{1}{2}\int_{\gamma}pdx-xdp=-\int_{\gamma}pdx
\]
and hence%
\[
\psi(x,T)=\exp\left(  -\tfrac{i}{\hbar}%
{\textstyle\int_{\gamma}}
pdx\right)  \psi(x,0);
\]
the initial and final states are the same up to a phase factor; this the
occurrence of a geometric phase shift, which is a phase difference acquired
over the course of a cycle (\textquotedblleft Berry's phase\textquotedblright%
\ \cite{Berry}).

Propositions \ref{prop9} and \ref{prop10} will be combined in Section
\ref{subsech} to explicitly solve the Schr\"{o}dinger equation associated with
Hamiltonians of the type%
\[
H(z,t)=\frac{1}{2}M(t)z^{2}+z_{0}(t)\cdot z.
\]
As noted in Burdet \textit{et al}. \cite{bu}, \S 6, there are difficulties in
applying directly the path lifting property, but we will derive the result
very simply, using a factorization result for evolution propagators.

\subsection{$\operatorname*{Mp}(n)$ and Feichtinger's algebra}

\subsubsection{The Feichtinger algebra $S_{0}(\mathbb{R}^{n})$}

The metaplectic group $\operatorname*{Mp}(n)$, and its inhomogeneous extension
$\operatorname*{IMp}(n)$, consist of unitary operators on $L^{2}%
(\mathbb{R}^{n})$; it is quite usual to perform practical calculations using
the Schwartz space $\mathcal{S}(\mathbb{R}^{n})$ which is dense in
$L^{2}(\mathbb{R}^{n})$, and it is has become quite usual in quantum theory to
use the latter as a natural \textquotedblleft reservoir\textquotedblright\ for
wavepackets. This choice is however inconvenient, and this for two reasons.
The first is of a mathematical nature: $\mathcal{S}(\mathbb{R}^{n})$ is a
Fr\'{e}chet space, defined by an infinite number of seminorms, and this can
make the proof of continuity properties for operators quite complicated. The
second inconvenience is physical: the elements of $\mathcal{S}(\mathbb{R}%
^{n})$ are $C^{\infty}$ functions, and this is very restrictive from a
physical point of view since its excludes \textit{de facto} many realistic
wavepackets (sharp pulses, step functions, etc.) such that, for instance,
\[
\psi(x)=\left\{
\begin{array}
[c]{c}%
1-|x|\text{ \textit{if} }|x|\leq1\\
0\text{ \textit{if} }|x|>1
\end{array}
.\right.
\]
There is however an alternative choice to $\mathcal{S}(\mathbb{R}^{n})$ which
has none of the shortcomings above. It consists in using the Feichtinger
algebra $S_{0}(\mathbb{R}^{n})$ (and its generalizations, the modulation
spaces $M_{s}^{q}(\mathbb{R}^{n})$), which was introduced by Feichtinger
\cite{hanskernel,fe81,fe81bis} in order to deal with localization problems in
functional analysis arising in signal theory and time-frequency analysis.
While signal theorists define $S_{0}(\mathbb{R}^{n})$ using \textquotedblleft
modulation operators\textquotedblright, it is better, keeping quantum
mechanics in mind, to recast its definition in terms of the cross-Wigner
transform%
\begin{equation}
W(\psi,\phi)(z)=\left(  \tfrac{1}{2\pi\hbar}\right)  ^{n}\int_{\mathbb{R}^{n}%
}e^{-\frac{i}{\hbar}p\cdot y}\psi(x+\tfrac{1}{2}y)\phi^{\ast}(x-\tfrac{1}%
{2}y)dy.
\end{equation}

\begin{definition}
\label{defmodspw}The Feichtinger algebra $S_{0}(\mathbb{R}^{n})$ consists of
all $\psi\in\mathcal{S}^{\prime}(\mathbb{R}^{n})$ such that $W(\psi,\phi)\in
L^{1}(\mathbb{R}^{2n})$ for every $\phi\in\mathcal{S}(\mathbb{R}^{n})$,
$\phi\neq0$. The number%
\begin{equation}
||\psi||_{\phi,S_{0}}=||W(\psi,\phi)||_{L^{1}}=\int_{\mathbb{R}^{2n}}%
|W(\psi,\phi)(z)|dz \label{wignorm}%
\end{equation}
is called the\ Wigner norm of $\psi$ relative to the window $\phi$.
\end{definition}

A striking, but not immediately obvious, fact is that all the norms
(\ref{wignorm}) are equivalent when $\phi$ ranges over $\mathcal{S}%
(\mathbb{R}^{n})$, and define a Banach space structure on the Feichtinger
algebra. The \textquotedblleft windows\textquotedblright\ used in the
definition of $S_{0}(\mathbb{R}^{n})$ can themselves be chosen in
$S_{0}(\mathbb{R}^{n})$:

\begin{proposition}
\label{prop7}Let both $\psi$ and $\phi$ be in $L^{2}(\mathbb{R}^{n})$. (i) If
$W(\psi,\phi)\in L^{1}(\mathbb{R}^{2n})$ then both $\psi$ and $\phi$ are in
$S_{0}(\mathbb{R}^{n})$; (ii) We have $\psi\in S_{0}(\mathbb{R}^{n})$ if and
only if $W(\psi,\phi)\in L^{1}(\mathbb{R}^{2n})$ for one (and hence every)
$\phi\in S_{0}(\mathbb{R}^{n})$.
\end{proposition}

It immediately follows from this characterization of $S_{0}(\mathbb{R}^{n})$ that:

\begin{corollary}
A function $\psi\in L^{2}(\mathbb{R}^{n})$ belongs to $S_{0}(\mathbb{R}^{n})$
if and only if $W\psi\in L^{1}(\mathbb{R}^{2n})$.
\end{corollary}

\begin{proof}
In view of the statement (i) in the Proposition above the condition $W\psi\in
L^{1}(\mathbb{R}^{2n})$ implies that $\psi\in S_{0}(\mathbb{R}^{n})$. If
conversely $\psi\in S_{0}(\mathbb{R}^{n})$ then $W\psi\in L^{1}(\mathbb{R}%
^{2n})$ in view of the statement (ii) in the same Proposition.
\end{proof}

The elements of $S_{0}(\mathbb{R}^{n})$ are continuous functions; in fact:%

\begin{equation}
S_{0}(\mathbb{R}^{n})\subset C^{0}(\mathbb{R}^{n})\cap L^{1}(\mathbb{R}%
^{n})\cap F(L^{1}(\mathbb{R}^{n})). \label{inclo}%
\end{equation}
It easily follows from Riemann--Lebesgue's lemma that each $\psi\in
S_{0}(\mathbb{R}^{n})$ is bounded and that we have $\lim_{|z|\rightarrow
\infty}\psi=0$. The Feichtinger algebra is a Banach algebra, both for
pointwise product and for convolution. In fact, $\psi\in L^{1}(\mathbb{R}%
^{n})$ and $\psi^{\prime}\in S_{0}(\mathbb{R}^{n})$. Then $\psi\ast
\psi^{\prime}\in S_{0}(\mathbb{R}^{n})$ and we have%
\begin{equation}
||\psi\ast\psi^{\prime}||_{\phi,S_{0}}\leq||\psi||_{L^{1}}||\psi^{\prime
}||_{\phi,S_{0}} \label{starlm}%
\end{equation}
for every window $\phi\in\mathcal{S}(\mathbb{R}^{n})$. Thus, if $\psi\in
L^{1}(\mathbb{R}^{n})$ and $\psi^{\prime}\in S_{0}(\mathbb{R}^{n})$ then
$\psi\ast\psi^{\prime}\in S_{0}(\mathbb{R}^{n})$, so that
\[
L^{1}(\mathbb{R}^{n})\ast S_{0}(\mathbb{R}^{n})\subset S_{0}(\mathbb{R}^{n}).
\]
That $S_{0}(\mathbb{R}^{n})$ also is closed under pointwise multiplication
follows, taking Fourier transforms and using the fact that $\psi\in
S_{0}(\mathbb{R}^{n})$ if and only if $F\psi\in S_{0}(\mathbb{R}^{n})$.

\subsubsection{Application to Schr\"{o}dinger's equation}

The fact that $F(S_{0}(\mathbb{R}^{n}))=S_{0}(\mathbb{R}^{n})$ is a particular
case of the following essential property:

\begin{proposition}
\label{prop8}The Feichtinger algebra $S_{0}(\mathbb{R}^{n})$ is closed under
the action of the inhomogeneous metaplectic group $\operatorname*{IMp}(n)$.
\end{proposition}

\begin{proof}
We recall the two following properties of the cross-Wigner transform
\cite{Birk,Birkbis}: for every $\widehat{S}\in\operatorname*{Mp}(n)$ we have%
\[
W(\widehat{S}\psi,\widehat{S}\phi)(z)=W(\psi,\phi)(S^{-1}z)
\]
for $S=\Pi(\widehat{S})$; for every $z_{0}\in\mathbb{R}^{2n}$%
\[
W(\widehat{T}(z_{0})\psi,\widehat{T}(z_{0})\phi)(z)=W(\psi,\phi)(z-z_{0}).
\]
Since $\operatorname*{IMp}(n)$ is generated by the operators $\widehat{T}%
(z_{0})$ and $\widehat{S}\in\operatorname*{Mp}(n)$ it suffices to show that if
$\psi\in S_{0}(\mathbb{R}^{n})$ then $\widehat{T}(z_{0})\psi\in S_{0}%
(\mathbb{R}^{n})$ and $\widehat{S}\psi\in S_{0}(\mathbb{R}^{n})$. Taking
definition (\ref{wignorm}) into account we have%
\begin{align*}
||\widehat{T}(z_{0})\psi||_{\phi,S_{0}}  &  =\int_{\mathbb{R}^{2n}%
}|W(\widehat{T}(z_{0})\psi,\phi)(z)|dz\\
&  =\int_{\mathbb{R}^{2n}}|W(\widehat{T}(z_{0})\psi,\widehat{T}(-z_{0}%
)\phi)(z-z_{0})|dz\\
&  =||\psi||_{\widehat{T}(-z_{0})\phi,S_{0}}%
\end{align*}
hence $\widehat{T}(z_{0})\psi\in S_{0}(\mathbb{R}^{n})$ since $\widehat{T}%
(-z_{0})\phi\in\mathcal{S}(\mathbb{R}^{n})$. Similarly,
\begin{align*}
||\widehat{S}\psi||_{\phi,S_{0}}  &  =\int_{\mathbb{R}^{2n}}|W(\widehat{S}%
\psi,\phi)(z)|dz\\
&  =\int_{\mathbb{R}^{2n}}|W(\psi,\widehat{S}^{-1}\phi)(S^{-1}z)|dz\\
&  =||\psi||_{\widehat{S}^{-1}\phi,S_{0}}%
\end{align*}
and $\widehat{S}^{-1}\phi\in\mathcal{S}(\mathbb{R}^{n})$ hence $\widehat{S}%
\psi\in S_{0}(\mathbb{R}^{n})$ since the Wigner norms $||\cdot||_{\phi,S_{0}}$
and $||\cdot||_{\widehat{S}^{-1}\phi,S_{0}}$ are equivalent.
\end{proof}

It turns out (but we do not prove it here, see \cite{gr01,Birkbis}) that
$S_{0}(\mathbb{R}^{n})$ is the \emph{smallest Banach algebra} containing the
Schwartz class $\mathcal{S}(\mathbb{R}^{n})$ which is closed under the action
of the inhomogeneous group $\operatorname*{IMp}(n)$.

An immediate important application of Proposition \ref{prop8} is:

\begin{corollary}
Let $H$ be a Hamiltonian function of the type%
\[
H(z,t)=\frac{1}{2}M(t)z^{2}%
\]
where $M(t)\in\operatorname*{Sym}(2n,\mathbb{R})$ is a $C^{j}$ ($j\geq2$)
function of $t\in\mathbb{R}$. Let $\psi_{0}\in S_{0}(\mathbb{R}^{n})$. Let
$\psi$ be the solution of Schr\"{o}dinger's equation%
\[
i\hbar\frac{\partial\psi}{\partial t}=\widehat{H}\psi\text{ \ , \ }\psi
(\cdot,0)=\psi_{0}.
\]
The function $\psi(\cdot,t)$ belongs to $S_{0}(\mathbb{R}^{n})$ for every
$t\in\mathbb{R}$.
\end{corollary}

\begin{proof}
It immediately follows from Propositions \ref{prop9} and \ref{prop8}.
\end{proof}

This result will be extended to inhomogeneous quadratic Hamiltonians in
Section \ref{subsech}.

\subsection{The Weyl symbol of a metaplectic operator}

Metaplectic operators map $\mathcal{S}(\mathbb{R}^{n})\longrightarrow
\mathcal{S}(\mathbb{R}^{n})$ and it therefore makes sense to speak about their
Weyl symbol. (We will review in some detail the notion of Weyl symbol of an
operator in Section \ref{secweyl}, to which we refer for the definitions used here.)

\subsubsection{The symplectic Cayley transform}

Let $\operatorname*{Sp}_{0}(n)$ be the closed subset of $\operatorname*{Sp}%
(n)$ defined by%
\[
\operatorname*{Sp}\nolimits_{0}(n)=\{S\in\operatorname*{Sp}(n):\det
(S-I)\neq0\}.
\]
We have proven in \cite{go05,go07,gogo06} that:

\begin{proposition}
\label{prop3}(i) The mapping $\Phi:S\longmapsto M_{S}$ defined by%
\[
\Phi(S)=M_{S}=\frac{1}{2}J(S+I)(S-I)^{-1}%
\]
is an injection of $\operatorname*{Sp}\nolimits_{0}(n)$ into the set
$\operatorname*{Sym}(2n,\mathbb{R})$ of real symmetric $2n\times2n$ matrices.
(ii)The inverse of $\Phi$ is given by%
\[
\Phi^{-1}(M)=\left(  M-\tfrac{1}{2}J\right)  ^{-1}\left(  M+\tfrac{1}%
{2}J\right)
\]
and we have $\Phi(S^{-1})=-\Phi(S)$.
\end{proposition}

We call the mapping $\Phi$ the \textit{symplectic Cayley transform} and
$M_{S}$ the Cayley matrix associated with $S\in\operatorname*{Sp}(n)$.

\subsubsection{The operators $\protect\widehat{R}_{\nu}(S)$}

Let now $\nu$ be an arbitrary real number and define, for $S\in
\operatorname*{Sp}\nolimits_{0}(n)$, an operator $\widehat{R}_{\nu}(S)$ by the
formula%
\begin{equation}
\widehat{R}_{\nu}(S)=\left(  \tfrac{1}{2\pi\hbar}\right)  ^{n}i^{\nu}%
\sqrt{|\det(S-I)|}\int_{\mathbb{R}^{2n}}\widehat{T}(Sz)\widehat{T}(-z)dz.
\label{Rvs1}%
\end{equation}
Using the formula (\ref{hw3}) it is easy to rewrite (\ref{Rvs1}) as%
\begin{equation}
\widehat{R}_{\nu}(S)=\left(  \frac{1}{2\pi\hbar}\right)  ^{n}\frac{i^{\nu}%
}{\sqrt{|\det(S-I)|}}\int_{\mathbb{R}^{2n}}e^{\frac{i}{2\hbar}\Phi(S)z^{2}%
}\widehat{T}(z)dz. \label{Rvs2}%
\end{equation}
We have proven in \cite{go05} (also \cite{go07,gogo06}) the following result:

\begin{proposition}
\label{prop4}(i) Let $\widehat{S}_{W,m}\in\operatorname*{Mp}(n)$ be such that
$\Pi^{\hbar}(\widehat{S}_{W,m})=S_{W}\in\operatorname*{Sp}\nolimits_{0}(n)$.
We have%
\begin{equation}
\widehat{S}_{W,m}=\widehat{R}_{m-\operatorname*{Inert}W_{xx}}(S_{W})
\label{Rvs3}%
\end{equation}
where $\operatorname*{Inert}W_{xx}$ is the index of inertia of the quadratic
form $x\longmapsto W(x,x)$. (ii) Every $\widehat{S}\in\operatorname*{Mp}(n)$
can be written as a product $\widehat{S}_{W,m}\widehat{S}_{W^{\prime
},m^{\prime}}$ with $S_{W},S_{W^{\prime}}\in\operatorname*{Sp}\nolimits_{0}%
(n)$ and if $S=\Pi^{\hbar}(\widehat{S})\in\operatorname*{Sp}\nolimits_{0}(n)$
then%
\begin{equation}
\widehat{S}=\widehat{R}_{m+m^{\prime}+\frac{1}{2}\operatorname*{sign}%
(\Phi(S)+\Phi(S^{\prime}))}(SS^{\prime}) \label{Rvs4}%
\end{equation}
where $\operatorname*{sign}(M)$ is the signature of the symmetric matrix $M$.
\end{proposition}

\begin{remark}
Formula (\ref{Rvs3}) identifies, modulo $4$, the integer $\nu
=m-\operatorname*{Inert}W_{xx}$ with the Conley--Zehnder index of a path
joining the identity $I$ to $S$ in $\operatorname*{Sp}(n)$; see de Gosson
\cite{go09,Birk}.
\end{remark}

Proposition \ref{prop4} has the following consequences:

\begin{corollary}
(i) The twisted Weyl symbol $(s_{W,m})_{\sigma}$ of $\widehat{S}_{W,m}%
\in\operatorname*{Mp}(n)$ ($S_{W}\in\operatorname*{Sp}\nolimits_{0}(n)$) is
given by%
\begin{equation}
(s_{W,m})_{\sigma}(z)=\frac{i^{m-\operatorname*{Inert}W_{xx}}}{\sqrt
{|\det(S_{W}-I)|}}e^{\frac{i}{2\hbar}\Phi(S_{W})z^{2}}. \label{weyl1}%
\end{equation}
(ii) Assuming in addition that $\det(S_{W}+I)\neq0$ the Weyl symbol $s_{W}$ of
$\widehat{S}_{W,m}$ is given by%
\begin{equation}
s_{W,m}(z)=2^{n/2}\frac{i^{m-\operatorname*{Inert}W_{xx}}}{\sqrt{|\det
(S_{W}+I)|}}e^{-\frac{i}{2\hbar}\Phi(S_{W}^{-1})Jz\cdot Jz} \label{weyl1bis}%
\end{equation}

\end{corollary}

\begin{proof}
Formula (\ref{weyl1}) is an immediate consequence of (\ref{Rvs2}) and
(\ref{Rvs3}) in view of (\ref{Opa}). To prove formula (\ref{weyl1bis}) it
suffices to note that $s_{W,m}$ is the symplectic Fourier transform
(\ref{sympFT}) of $(s_{W,m})_{\sigma}$ (because the latter is involutive) and
to apply the well-known Fresnel formula giving the Fourier \ transform of a
Gaussian (see \cite{Birk}, \S 7.4).
\end{proof}

The case of the inhomogeneous metaplectic group easily follows:

\begin{corollary}
The twisted Weyl symbol of $\widehat{T}(z_{0})\widehat{S}_{W,m}\in
\operatorname*{IMp}(n)$ (with $S_{W}\in\operatorname*{Sp}\nolimits_{0}(n)$) is
given by%
\begin{equation}
a_{\sigma}(z)=\frac{i^{m-\operatorname*{Inert}W_{xx}}}{\sqrt{|\det(S-I)|}%
}e^{\frac{i}{2\hbar}[\Phi(S)(z-z_{0})^{2}-\sigma(z,z_{0})]}. \label{weyl2}%
\end{equation}

\end{corollary}

\begin{proof}
Using successively (\ref{Rvs2}) and (\ref{hw3}) we get%
\begin{align*}
\widehat{T}(z_{0})\widehat{S}_{W,m}  &  =\left(  \frac{1}{2\pi\hbar}\right)
^{n}\frac{i^{\nu}}{\sqrt{|\det(S-I)|}}\int_{\mathbb{R}^{2n}}e^{\frac{i}%
{2\hbar}\Phi(S)z^{2}}\widehat{T}(z_{0})\widehat{T}(z)dz\\
&  =\left(  \frac{1}{2\pi\hbar}\right)  ^{n}\frac{i^{\nu}}{\sqrt{|\det(S-I)|}%
}\int_{\mathbb{R}^{2n}}e^{\frac{i}{2\hbar}\Phi(S)z^{2}}e^{\frac{i}{2\hbar
}\sigma(z_{0},z)}\widehat{T}(z_{0}+z)dz;
\end{align*}
formula (\ref{weyl2}) follows by the change of variables $z_{0}+z\longmapsto
z$.
\end{proof}

The Weyl symbol of a metaplectic operator $\widehat{S}$ such that
$\det(S-I)=0$ cannot be determined using the methods above; one has to use
direct methods.

\begin{example}
Assume $S_{t}=%
\begin{pmatrix}
1 & t\\
0 & 1
\end{pmatrix}
$ for $t\in\mathbb{R}$. We have $\det(S_{t}-I)=0$ hence the formulas above do
not apply. However, using the kernel formula (\ref{ker}) a direct calculation
using Fresnel integrals leads to the following particularly simple expressions
for the Weyl symbol $s_{t}$ of $\widehat{S}_{t}$ and of its symplectic Fourier
transform $(s_{t})_{\sigma}$:%
\begin{equation}
s_{t}(x,p)=e^{-\frac{i}{2\hbar}p^{2}t\text{ \ }}\text{and \ }(s_{t})_{\sigma
}(x,p)=\pi i^{-1/2}\sqrt{\frac{2\hbar}{t}}e^{-\frac{i}{2\hbar t}x^{2}}%
\otimes\delta(p). \label{weylfree}%
\end{equation}

\end{example}

\section{Quantum Propagators and Isotopies\label{sec3}}

Schr\"{o}dinger's equation is the quantum analogue of Hamilton's equations;
symplectomorphisms are replaced with unitary operators obtained from the
Hamiltonian function by a quantization procedure. The process can be reversed,
by \textit{dequantizing} these unitary operators. This is consistent with
Mackey's view \cite{ma98} following which quantum mechanics is a refinement of
Hamiltonian mechanics; (also see Marsden \cite{Marsden}, Emch \cite{em83}).
Mackey added that dequantization is a more fundamental process than quantization.

\subsection{The problem of quantization\label{secweyl}}

The problem of how to \textquotedblleft quantize\textquotedblright\ properly a
\textquotedblleft classical observable\textquotedblright\ is still an open
one. Mathematically, it amounts to finding a pseudo-differential calculus
suitable for applications to physical problems.

\subsubsection{Shubin and Weyl correspondence}

Let us review some basic concepts from pseudo-differential theory following
Shubin \cite{sh87}. Let $A$ be a continuous operator $\mathcal{S}%
(\mathbb{R}^{n})\longrightarrow\mathcal{S}^{\prime}(\mathbb{R}^{n})$; in view
of Schwartz's kernel theorem \cite{gr06} there exists a distribution
$K\in\mathcal{S}^{\prime}(\mathbb{R}^{n}\times\mathbb{R}^{n})$ such that%
\[
\langle A\psi,\phi\rangle=\langle\langle K,\psi\otimes\phi\rangle\rangle.
\]
($\langle\cdot,\cdot\rangle$ and $\langle\langle\cdot,\cdot\rangle\rangle$ the
distributional brackets on $\mathbb{R}^{n}$ and $\mathbb{R}^{n}\times
\mathbb{R}^{n}$, respectively). Writing formally%
\[
A\psi(x)=\int_{\mathbb{R}^{n}}K(x,y)\psi(y)dy
\]
Let now $\tau$ be an arbitrary fixed real number; the $\tau$-symbol $a_{\tau
}\in\mathcal{S}^{\prime}(R^{2n})$ of $A$ is defined by the Fourier transform%
\begin{equation}
a_{\tau}(x,p)=\int_{\mathbb{R}^{n}}e^{-\frac{i}{\hbar}p\cdot y}K_{A}(x+\tau
y,x-(1-\tau)y)dy. \label{atau}%
\end{equation}
One shows \cite{Birk,Birkbis,sh87} that $A$ can be (formally) written as a
pseudo-differential operator
\begin{equation}
A\psi(x)=\left(  \tfrac{1}{2\pi\hbar}\right)  ^{n}\iint\nolimits_{\mathbb{R}%
^{2n}}e^{\frac{i}{\hbar}p\cdot(x-y)}a_{\tau}((1-\tau)x+\tau y,p)\psi(y)dydp.
\label{taupdo}%
\end{equation}
We will write $A=\operatorname*{Op}^{\tau}(a_{\tau})$ and call $A$ the $\tau
$-operator with symbol $a_{\tau}$. In particular, the choice $\tau=\frac{1}%
{2}$ corresponds to the Weyl operator $\widehat{A}=\operatorname*{Op}%
^{w}(a_{1/2})$; its Weyl symbol is the distribution $a\in\mathcal{S}^{\prime
}(\mathbb{R}^{2n})$ given by
\begin{equation}
a(x,p)=\int_{\mathbb{R}^{n}}e^{-\frac{i}{\hbar}p\cdot y}K(x+\tfrac{1}%
{2}y,x-\tfrac{1}{2}y)dy \label{ker}%
\end{equation}
(see \textit{e.g.} Shubin \cite{sh87}, de Gosson \cite{Birk,Birkbis}). We will
write $\widehat{A}=\operatorname*{Op}^{w}(a)$; the operator $\widehat{A}$ is
then given by%
\begin{equation}
\widehat{A}\psi(x)=\left(  \tfrac{1}{2\pi\hbar}\right)  ^{n}\int%
_{\mathbb{R}^{2n}}e^{-\frac{i}{\hbar}p\cdot y}a(\tfrac{1}{2}(x+y),p)\psi
(y)dpdy; \label{Weyl}%
\end{equation}
one proves that $\widehat{A}$ is self-adjoint if its symbol $a$ is real (this
property is not shared by the operators (\ref{taupdo}) for $\tau\neq\frac
{1}{2}$).

One shows that the following harmonic representation of $\widehat{A}$ holds:%
\begin{equation}
\operatorname*{Op}\nolimits^{w}(a)=\left(  \tfrac{1}{2\pi\hbar}\right)
^{n}\int_{\mathbb{R}^{2n}}a_{\sigma}(z)\widehat{T}(z)dz \label{Opa}%
\end{equation}
where $a_{\sigma}$ is the symplectic Fourier transform of the symbol $a$;
formally%
\begin{equation}
a_{\sigma}(z)=\left(  \tfrac{1}{2\pi\hbar}\right)  ^{n}\int_{\mathbb{R}^{2n}%
}e^{-\frac{i}{\hbar}\sigma(z,z^{\prime})}a(z^{\prime})dz^{\prime}.
\label{sympFT}%
\end{equation}
The distribution $a_{\sigma}$ is called the twisted (or covariant) Weyl symbol
of $\widehat{A}$.

A characteristic property of Weyl operators is their symplectic covariance; it
is the quantum analogue of the property (\ref{FHg}) of Hamilton's equations.

\begin{proposition}
For every $\widehat{S}\in\operatorname*{Mp}(n)$ with $S=\Pi(\widehat{S})$ we
have
\begin{equation}
\widehat{S}\operatorname*{Op}\nolimits^{w}(a)\widehat{S}^{-1}%
=\operatorname*{Op}\nolimits^{w}(a\circ S^{-1}). \label{sympco2}%
\end{equation}
Conversely, if $A=\operatorname*{Op}^{\tau}(a_{\tau})$ is such that
$\widehat{S}\operatorname*{Op}^{\tau}(a_{\tau})\widehat{S}^{-1}%
=\operatorname*{Op}^{\tau}(a_{\tau}\circ S^{-1})$ then we must have
$\tau=\frac{1}{2}$; i.e. $A$ is a Weyl operator.
\end{proposition}

A similar property does not hold for arbitrary Shubin $\tau$-operators. See
however de Gosson \cite{go13} for partial symplectic covariance results.

\subsubsection{Moyal product and twisted convolution}

The Moyal (or star) product $a\ast_{\hbar}b$ of $a,b\in\mathcal{S}%
(\mathbb{R}^{2n})$ is defined by
\begin{equation}
a\ast_{\hbar}b(z)=\left(  \tfrac{1}{4\pi\hbar}\right)  ^{2n}\int%
_{\mathbb{R}^{2n}}e^{\frac{i}{2\hbar}\sigma(z^{\prime},z^{\prime\prime}%
)}a(z+\tfrac{1}{2}z^{\prime})b(z-\tfrac{1}{2}z^{\prime\prime})dz^{\prime
}dz^{\prime\prime}; \label{Moyal1}%
\end{equation}
this formula can be alternatively written%
\begin{equation}
a\ast_{\hbar}b(z)=\left(  \tfrac{1}{4\pi\hbar}\right)  ^{2n}\int%
_{\mathbb{R}^{4n}}e^{-\frac{2i}{\hbar}\partial\sigma(z,z^{\prime}%
,z^{\prime\prime})}a(z^{\prime})b(z^{\prime\prime})dz^{\prime}dz^{\prime
\prime}; \label{Moyal2}%
\end{equation}
where $\partial\sigma$ is the coboundary of the symplectic form:
\[
\partial\sigma(z,z^{\prime},z^{\prime\prime})=\sigma(z,z^{\prime}%
)-\sigma(z,z^{\prime\prime})+\sigma(z^{\prime},z^{\prime\prime}).
\]

The twisted convolution $a\#b$ i defined by%
\begin{equation}
a\#b=F_{\sigma}(F_{\sigma}a\ast_{\hbar}F_{\sigma}b); \label{twist1}%
\end{equation}
explicitly:%
\begin{equation}
a\#b(z)=\left(  \tfrac{1}{2\pi\hbar}\right)  ^{n}\int_{\mathbb{R}^{2n}%
}e^{\frac{i}{2\hbar}\partial\sigma(z,z^{\prime},z^{\prime\prime}%
)}a(z-z^{\prime})b(z^{\prime})dz^{\prime} \label{twist2}%
\end{equation}
The relation of these notions with Weyl operator theory is the following
\cite{Birk,Birkbis,Hoermander}:

\begin{proposition}
Let $\widehat{A}=\operatorname*{Op}^{w}(a)$, $\widehat{B}=\operatorname*{Op}%
^{w}(b)$ and $\widehat{C}=\widehat{A}\widehat{B}$. We have $\widehat{C}%
=\operatorname*{Op}^{w}(a\ast_{\hbar}b)$ and $(a\ast_{\hbar}b)_{\sigma}=a\#b$.
\end{proposition}

The papers by Hansen \cite{Hansen} and Littlejohn \cite{Littlejohn} contain
nice reviews of the topic; also see the evergreen book by Folland \cite{fo89},
and de Gosson \cite{Birk,Birkbis} where the closely related topic of
deformation quantization is also discussed.

\subsubsection{Born--Jordan quantization}

It is the property of symplectic covariance (\ref{sympco2}), and the fact that
$\widehat{A}=\operatorname*{Op}^{w}(a)$ is self-adjoint if $a$ is real, which
are at the origin of the privileged role played by Weyl operators, especially
in quantum mechanics. The rub comes from the fact that it is not quite obvious
that Weyl quantization is really the right choice in physics; for instance we
have shown in \cite{go14} that if one wants the Schr\"{o}dinger and Heisenberg
pictures of quantum mechanics to be equivalent (which is assumed to be true by
most physicists), then one has to use the so-called \textit{Born--Jordan
quantization} scheme \cite{bj} we have studied in \cite{go13,golu11}; the
latter is obtained by averaging the $\tau$-symbol (\ref{atau}) over the
interval $[0,1]$:
\begin{subequations}
\label{0}%
\begin{equation}
a_{\mathrm{BJ}}(x,p)=\int_{0}^{1}a_{\tau}(x,p)d\tau. \label{abj}%
\end{equation}

In the case of polynomials, the distinction between Weyl and Born--Jordan
quantization is known since the early years of quantum mechanics; for instance
the Weyl correspondence is the prescription (in dimension $n=1$)%
\end{subequations}
\begin{equation}
x^{s}q^{r}\overset{\mathrm{Weyl}}{\longrightarrow}\frac{1}{2^{s}}\sum_{\ell
=0}^{s}%
\begin{pmatrix}
s\\
\ell
\end{pmatrix}
\widehat{p}^{s-\ell}\widehat{x}^{r}\widehat{p}^{\ell} \label{w2}%
\end{equation}
while Born--Jordan imposes the apparently less weighted rule%
\begin{equation}
x^{s}q^{r}\overset{\mathrm{BJ}}{\longrightarrow}\frac{1}{s+1}\sum_{\ell=0}%
^{s}\widehat{p}^{s-\ell}\widehat{x}^{r}\widehat{p}^{\ell} \label{bj1}%
\end{equation}
It turns out, however, that Weyl quantization and Born--Jordan quantization
coincide for large classes of operators; for instance for all physically
interesting operators associated with Hamiltonians of the type $T(p)+V(x)$ or,
more generally
\[
H=\sum_{j=1}^{n}\frac{1}{2m_{j}}(p_{j}-A_{j}(x,t))^{2}+V(x,t)
\]
(see de Gosson \cite{go13}), so one can safely conclude that within the
framework of the present paper it really doesn't matter whether one uses the
Weyl or Born--Jordan quantization rules.

\subsection{Paths of unitary operators}

Hamilton's equation govern the time evolution of classical systems; in
non-relativistic quantum mechanics Schr\"{o}dinger's equation plays a similar
role. While points in the classical phase space are propagated using canonical
transformation, the evolution of quantum-mechanical wavefunctions is obtained
using paths of unitary operators.

\subsubsection{Schr\"{o}dinger's equation}

Let $\widehat{H}$ be the Weyl quantization of a Hamiltonian function $H$; in
physics $H$ is typically of the type
\[
H=\sum_{j=1}^{n}\frac{1}{2m_{j}}(p_{j}-A_{j}(x,t))^{2}+V(x,t)
\]
where the potential functions $A_{j}$ and $V$ satisfy some smoothness
conditions. When $H$ is of the type above this operator is given by%
\[
\widehat{H}=-\sum_{j=1}^{n}\frac{1}{2m_{j}}\left(  -i\hbar\frac{\partial
}{\partial x_{j}}-A_{j}(x,t)\right)  ^{2}+V(x,t)
\]
in both the Weyl and Born--Jordan quantization schemes. The Schr\"{o}dinger
equation corresponding to the Hamiltonian function $H$ is
\begin{equation}
i\hbar\frac{\partial\psi}{\partial t}(x,t)=\widehat{H}\psi(x,t)\text{ , }%
\psi(\cdot,t^{\prime})=\psi^{\prime} \label{Sch1}%
\end{equation}
where the initial Cauchy datum $\psi^{\prime}$ is a function (or distribution)
belonging to some suitable function space (Yajima \cite{ya87,ya96} has
introduced a class of distributions which guarantees the existence of the
solutions). Assuming that the solution $\psi$ of (\ref{Sch1}) exists and is
unique for $t$ in some interval $I=[t^{\prime}-T,t^{\prime}+T]$ we can write
$\psi=U_{t,t^{\prime}}^{H}\psi^{\prime}$ where $U_{t,t^{\prime}}^{H}$ (the
evolution operator, or propagator) is a unitary operator on $L^{2}%
(\mathbb{R}^{n})$. The Chapman--Kolmogorov property%
\begin{equation}
U_{t,t^{\prime}}^{H}U_{t^{\prime},t^{\prime\prime}}^{H}=U_{t,t^{\prime\prime}%
}^{H}\text{ \ , \ }U_{t,t}^{H}=I_{\mathrm{d}} \label{chko}%
\end{equation}
holds whenever $U_{t,t^{\prime}}^{H}U_{t^{\prime},t^{\prime\prime}}^{H}$
exists. When $t^{\prime}=0$ we write $\psi^{\prime}=\psi_{0}$, $I_{T}=[-T,T]$,
and $U_{t,0}^{H}=U_{t}^{H}$. Notice that the semigroup property $U_{t}%
^{H}U_{t^{\prime}}^{H}=U_{t+t^{\prime}}^{H}$ only makes sense when $H$ is time-independent.

Since $U_{t,t^{\prime}}^{H}$ is a continuous operator $\mathcal{S}%
(\mathbb{R}^{n})\longrightarrow L^{2}(\mathbb{R}^{n})\subset\mathcal{S}%
^{\prime}(\mathbb{R}^{n})$, it is a bona fide Weyl operator; the determination
of its Weyl symbol is however in most cases a cumbersome exercise (to say the
least); Berezin and Shubin \cite{beshu} give an expression for the Weyl symbol
using a Feynman path-type integral, which is not very tractable in practice.

Suppose now that the Hamiltonian function is a (time-dependent) quadratic form
in the position and momentum variables; such a function can always be written
as
\begin{equation}
H(z,t)=\frac{1}{2}M(t)z^{2} \label{Hom}%
\end{equation}
where $t\longmapsto M(t)$ is a $C^{j}$ ($j\geq2$) mapping $\mathbb{R}%
\longrightarrow\operatorname*{Sym}(2n,\mathbb{R})$. In that case we have
$U_{t,t^{\prime}}^{H}=\widehat{S}_{t,t^{\prime}}$ where the two-parameter
family $(\widehat{S}_{t,t^{\prime}})$ is constructed as follows: setting
$U_{t}^{H}=U_{t,0}^{H}$ we have $U_{t,t^{\prime}}^{H}=U_{t}^{H}(U_{t^{\prime}%
}^{H})^{-1}$. Consider now the (time-dependent) Hamiltonian flow $(f_{t}%
^{H})_{t}$; it consists of a $C^{1}$ one-parameter family of linear mappings
$S_{t}\in\operatorname*{Sp}(n)$ passing through the identity at time $t=0$.
Using the path lifting property previously studied, to $(f_{t}^{H})_{t}%
=(S_{t})_{t}$ corresponds a unique family $(\widehat{S}_{t})_{t}$ of
metaplectic operators passing through the identity at time $t=0$ hence we have%
\begin{equation}
U_{t,t^{\prime}}^{H}=\widehat{S}_{t}(\widehat{S}_{t^{\prime}})^{-1}%
\in\operatorname*{Mp}(n). \label{stt}%
\end{equation}

\subsubsection{Quantum isotopies\label{sec4}}

In Proposition \ref{prop5} we showed that each symplectic isotopy is actually
a Hamiltonian flow. We are now going to prove a quantum analogue of this
property. We begin by defining the notion of \emph{quantum isotopy}, which is
the operator analogue of a symplectic isotopy.

\begin{definition}
A quantum isotopy is a $C^{1}$ one-parameter family $(U_{t})_{t}$ of unitary
operators on $L^{2}(\mathbb{R}^{n})$ having the following properties. (i)
$U_{0}$ is the identity operator: $U_{0}=I_{\mathrm{d}}$; (ii) There exists a
dense subspace $D$ of $L^{2}(\mathbb{R}^{n})$ such that%
\[
\left(  \frac{d}{dt}U_{t}\right)  \psi=\lim_{\Delta t\rightarrow0}%
\frac{U_{t+\Delta t}\psi-U_{t}\psi}{\Delta t}%
\]
exists for every $\psi\in D$. We call $D$ the domain of the quantum isotopy
$(U_{t})_{t}$. When $U_{t}\in\operatorname*{Mp}(n)$ for every $t\in\mathbb{R}$
we call $(U_{t})_{t}$ is a metaplectic isotopy.
\end{definition}

We will also use the following notation: for $t,t^{\prime}\in\mathbb{R}$ we
set $U_{t,t^{\prime}}=U_{t}U_{t^{\prime}}^{-1}$ hence $U_{t,0}=U_{t}$,
$U_{t,t}=I_{\mathrm{d}}$, and%
\begin{equation}
U_{t,t^{\prime}}U_{t^{\prime},t^{\prime\prime}}=U_{t,t^{\prime\prime}}\text{
\ , \ }(U_{t,t^{\prime}})^{-1}=U_{t^{\prime},t} \label{chako1}%
\end{equation}
(\textit{cf.} the Chapman--Kolmogorov law (\ref{chko})).

A quantum propagator $(U_{t}^{H})_{t}$ is \textit{de facto} a quantum isotopy:
we have $U_{t}^{H}=I_{\mathrm{d}}$ and the relation%
\[
\left(  \frac{d}{dt}U_{t}^{H}\right)  \psi=\frac{1}{i\hbar}\widehat{H}\psi
\]
holds for $\psi\in\mathcal{S}(\mathbb{R}^{n})$. Writing $\dot{U}_{t}=\frac
{d}{dt}U_{t}$ we have, more generally:

\begin{proposition}
Let $(U_{t})_{t}$ be a quantum isotopy with domain $D$. (i) The
(time-dependent) operator
\begin{equation}
\widehat{H}=i\hbar\dot{U}_{t}U_{t}^{-1} \label{hatdefine}%
\end{equation}
with domain $D_{\widehat{H}}=D$\ is self-adjoint; (ii) the function
$\psi=U_{t}\psi_{0}$ is a solution of Schr\"{o}dinger's equation%
\[
i\hbar\frac{\partial\psi}{\partial t}=\widehat{H}\psi\text{ \ , \ }\psi
(\cdot,0)=\psi_{0}%
\]
for every $\psi_{0}\in D_{\widehat{H}}$; (iii) When $(U_{t})_{t}$ is a
metaplectic isotopy $(\widehat{S}_{t})_{t}$, then $\widehat{H}$ is explicitly
given by
\begin{equation}
\widehat{H}=-\frac{1}{2}J\dot{S}_{t}S_{t}^{-1}\widehat{z}\cdot\widehat{z}
\label{hatxplicit}%
\end{equation}
where $(S_{t})_{t}$ is the symplectic isotopy obtained by projecting
$(\widehat{S}_{t})_{t}$ on $\operatorname*{Sp}(n)$.
\end{proposition}

\begin{proof}
(i) The operator $U_{t}$ is defined on $L^{2}(\mathbb{R}^{n})$ and the domain
of $\widehat{H}$ is hence that of $\dot{U}_{t}$. Let us set
\[
A_{\Delta t}=\frac{i\hbar}{\Delta t}(U_{t+\Delta t}-U_{t})U_{t}^{-1}%
=\frac{i\hbar}{\Delta t}(U_{t+\Delta t,0}-U_{t,0})U_{t,0}^{-1}%
\]
that is, since $(U_{t,0}^{-1})^{\ast}=U_{t}$, and using the rules
(\ref{chako1}),%
\begin{equation}
A_{\Delta t}=\frac{i\hbar}{\Delta t}(U_{t+\Delta t,0}-U_{t,0})U_{0,t}%
=\frac{i\hbar}{\Delta t}(U_{t+\Delta t,t}-I). \label{adt1}%
\end{equation}
It follows that for $\psi,\phi\in D$,%
\begin{equation}
\lim_{\Delta t\rightarrow0}\langle A_{\Delta t}\psi|\phi\rangle=i\hbar
\langle\dot{U}_{t}U_{t}^{-1}\psi|\phi\rangle=\langle\widehat{H}\psi
|\phi\rangle\text{.} \label{hosp1}%
\end{equation}
Similarly, taking into account the equality $U_{t+\Delta t,t}^{\ast
}=U_{t,t+\Delta t}$, and using again the rules (\ref{chako1}), the adjoint of
$A_{\Delta t}$ is given by
\[
A_{\Delta t}^{\ast}=-\frac{i\hbar}{\Delta t}(U_{t,t+\Delta t}-I)=-\frac
{i\hbar}{\Delta t}(U_{t,0}-U_{t+\Delta t,0})U_{0,t+\Delta t}%
\]
that is%
\begin{equation}
A_{\Delta t}^{\ast}=\frac{i\hbar}{\Delta t}(U_{t+\Delta t}-U_{t})U_{t+\Delta
t}^{-1}. \label{adt2}%
\end{equation}
It follows that for $\psi,\phi\in D$ we have%
\[
\lim_{\Delta t\rightarrow0}\langle\psi|A_{\Delta t}^{\ast}\phi\rangle
=i\hbar\langle\psi|\dot{U}_{t}U_{t}^{-1}\phi\rangle=\langle\psi|\widehat{H}%
\phi\rangle.
\]
Since $\langle A_{\Delta t}\psi|\phi\rangle=\langle\psi|A_{\Delta t}^{\ast
}\phi\rangle$ We have thus proven that $\langle\widehat{H}\psi|\phi
\rangle=\langle\psi|\widehat{H}\phi\rangle$ for all $\psi,\phi\in D$, which
shows the self-adjointness of the operator $\widehat{H}$. (ii) immediately
follows from the definition (\ref{hatdefine}) of $\widehat{H}$. (iii) Let us
set $S_{t}=\Pi(\widehat{S}_{t})$; the one-parameter family $(S_{t})_{t}$ is a
symplectic isotopy in $\operatorname*{Sp}(n)$, and is thus (Corollary
\ref{cor2}) the Hamiltonian flow determined by a quadratic Hamiltonian
function, the latter being given by formula (\ref{hamzo}):
\[
H(z,t)=-\frac{1}{2}J\dot{S}_{t}S_{t}^{-1}z\cdot z;
\]
the expression (\ref{hatxplicit}) of $\widehat{H}$ since $(\widehat{S}%
_{t})_{t}$ is just the lift to $\operatorname*{Mp}(n)$ of the symplectic
isotopy $(S_{t})_{t}$.
\end{proof}

Stone's theorem allows us to somewhat weaken the assumption on $(U_{t})_{t}$
when it is in addition a one-parameter group:

\begin{corollary}
Let $(U_{t})_{t}$ be a strongly continuous one-parameter group of unitary
operators on $L^{2}(\mathbb{R}^{n})$: $\lim_{t\rightarrow t_{0}}U_{t}%
\psi=U_{t_{0}}\psi$ for every $\psi\in L^{2}(\mathbb{R}^{n})$, and
$U_{t}U_{t^{\prime}}=U_{t+t^{\prime}}$ for all $t,t^{\prime}\in\mathbb{R}$.
Then $(U_{t})_{t}$ is a quantum isotopy, and we have $U_{t}=e^{-i\hbar
\widehat{H}/t}$.
\end{corollary}

\begin{proof}
By Stone's theorem \cite{Stone,AMR,RS} there exists a unique self-adjoint
operator $\widehat{H}$ (the infinitesimal generator), densely defined in
$L^{2}(\mathbb{R}^{n})$, whose domain contains $\mathcal{S}(\mathbb{R}^{n})$,
and such that $U_{t}=e^{-i\hbar\widehat{H}/t}$. This shows that the path
$t\longmapsto U_{t}$ is $C^{1}$ and hence a quantum isotopy since $i\hbar
\dot{U}_{t}U_{t}^{-1}=\widehat{H}$.
\end{proof}

\subsubsection{A factorization result\label{subsech}}

We are going to prove the analogue of Proposition \ref{prop6} about
Hamiltonian systems when $H_{0}$ is a quadratic polynomial (Weinstein
\cite{we85}) proves a similar result in the particular case $H_{1}=V(x,t)$).

\begin{proposition}
\label{propho}Suppose $H_{0}$ is a quadratic Hamiltonian of the type
(\ref{Hom}). Then the propagator $U_{t}^{H_{0}+H_{1}}$ is given by the formula%
\begin{equation}
U_{t}^{H_{0}+H_{1}}=U_{t}^{H_{0}}U_{t}^{H_{1}\circ S_{t}}=\widehat{S}_{t}%
U_{t}^{H_{1}\circ S_{t}} \label{ut}%
\end{equation}
where $(S_{t})_{t}=(f_{t}^{H_{0}})_{t}$ is the flow determined by the Hamilton
equations for $H_{0}$ (and thus $S_{t}\in\operatorname*{Sp}(n)$).
\end{proposition}

\begin{proof}
Let us set $V_{t}=U_{t}^{H_{0}}U_{t}^{H_{1}\circ S_{t}}$. We have, using
respectively the product rule and metaplectic covariance formula
(\ref{sympco2}) for Weyl operators,
\begin{align*}
i\hbar\frac{d}{dt}V_{t}  &  =\left(  i\hbar\frac{d}{dt}\widehat{S}_{t}\right)
U_{t}^{H_{1}\circ S_{t}}+\widehat{S}_{t}\left(  i\hbar\frac{d}{dt}U_{t}%
^{H_{1}\circ S_{t}}\right) \\
&  =\widehat{H}_{0}U_{t}^{H_{1}\circ S_{t}}+\widehat{S}_{t}\widehat{H_{1}\circ
S_{t}}U_{t}^{H_{1}\circ S_{t}}\\
&  =\widehat{H}_{0}U_{t}^{H_{1}\circ S_{t}}+\widehat{S}_{t}(\widehat{S}%
_{t}^{-1}\widehat{H}_{1}\widehat{S}_{t})U_{t}^{H_{1}\circ S_{t}}\\
&  =\widehat{H}_{0}U_{t}^{H_{1}\circ S_{t}}+\widehat{H}_{1}U_{t}^{H_{1}\circ
S_{t}}%
\end{align*}
which proves formula (\ref{ut}).
\end{proof}

The assumption that $H_{0}$ is quadratic is \emph{essential} in the proof of
formula (\ref{ut}), since the flow $(f_{t}^{H_{0}})$ then consists of
symplectic matrices, allowing us to use the symplectic covariance formula
$\widehat{H_{1}\circ S_{t}}=\widehat{S}_{t}^{-1}\widehat{H}_{1}\widehat{S}%
_{t}$. There is no analogue in the case of non-linear flows (see de Gosson
\cite{go11}). However, we conjecture that it is possible to obtain asymptotic
equalities in the limit $\hbar\rightarrow0$ (with arbitrary accuracy) using
the method developed by Lasser and her coworkers \cite{gala14,kela13}, and
which yield a refinement of Egorov's theorem \cite{eg69} on transformation
properties of pseudo-differential operators.

\begin{example}
\label{cor1}Assume that $H_{0}=\frac{1}{2}p^{2}$ and $H_{1}=V(x,t)$. Set
$H=H_{0}+H_{1}$. Then $U_{t}^{H}=\widehat{S}_{t}U_{t}^{V^{t}}$ where
$\widehat{S}_{t}$ is the free particle propagator%
\begin{equation}
\widehat{S}_{t}\psi(x)=\left(  \frac{1}{2\pi i\hbar|t|}\right)  ^{1/2}%
i^{m}\int_{-\infty}^{\infty}\exp\left[  \frac{i}{\hbar}\frac{(x-x^{\prime
})^{2}}{2t}\right]  \psi(x^{\prime})dx^{\prime} \label{free}%
\end{equation}
with $m=0$ for $t>0$ and $m=1$ for $t<0$, and $V_{t}(z,t)=V(x+pt,t)$ (the
latter being obtained by applying the path lifting method to $H_{0}$) since
the flow determined by $H_{0}$ consists of the linear mappings $S_{t}%
:(x,p)\longmapsto(x+pt,p)$.
\end{example}

As an illustration let us work out in detail the case of the harmonic
oscillator $H(x,p)=\frac{1}{2}(p^{2}+x^{2})$ considered in Examples \ref{exa1}
and \ref{exa3}. It gives rise to a group $(U_{t}^{H})_{t}$ of unitary
operators given, for $\psi\in\mathcal{S}(\mathbb{R})$ and $t\notin%
\pi\mathbb{Z}$, by the formula (\ref{sth}). Let us set $H_{0}(x,p)=\frac{1}%
{2}p^{2}$ and $H_{1}(x,p)=\frac{1}{2}x^{2}$. The evolution operator
$U_{t}^{H_{0}}$ is the quantization of the one-parameter group of $S_{t}=%
\begin{pmatrix}
1 & t\\
0 & 1
\end{pmatrix}
$ and thus consists of the metaplectic operators defined, for $\psi
\in\mathcal{S}(\mathbb{R})$ and $t\neq0$, by (\ref{free}). We have $H_{1}%
^{t}(x,p)=\frac{1}{2}(x+pt)^{2}$, and the evolution operator $U_{t}%
^{H_{1}\circ S_{t}}$ is thus the quantization of the linear flow (\ref{mat1}),
given by%
\[
f_{t}^{H_{1}^{t}}=%
\begin{pmatrix}
\cos t+t\sin t & \sin t-t\cos t\\
-\sin t & \cos t
\end{pmatrix}
;
\]
the symplectic matrices $f_{t}^{H_{1}^{t}}$ are free for $t\neq0$ and their
generating functions (\ref{W}) are%
\[
W(x,x^{\prime},t)=\frac{1}{\sin t-t\cos t}\left[  \frac{1}{2}x^{2}\cos
t-xx^{\prime}+\frac{1}{2}x^{\prime2}(\cos t+t\sin t)\right]  .
\]
It follows, in view of formulae (\ref{SWm}) and (\ref{WABCD}), that
$U_{t}^{H_{1}\circ S_{t}}$ consists of the metaplectic operators defined by%
\begin{equation}
U_{t}^{H_{1}\circ S_{t}}\psi(x)=\left(  \frac{1}{2\pi i\hbar}\right)
^{1/2}\frac{i^{m}}{\sqrt{|\sin t-t\cos t|}}\int_{-\infty}^{\infty}e^{\frac
{i}{\hbar}W(x,x^{\prime})}\psi(x^{\prime})dx^{\prime}%
\end{equation}
for an adequate choice of the Maslov index $m$. Now $U_{t}^{H_{0}}U_{t}%
^{H_{1}\circ S_{t}}$ is the product of two metaplectic operators of the type
(\ref{SWm}), namely $\widehat{S}_{W,m}\widehat{S}_{W^{\prime},m^{\prime}}$
with $P=L=Q=t^{-1}$ and
\[
P^{\prime}=\frac{\cos t}{\sin t-t\cos t}\text{ , }L^{\prime}=\frac{1}{\sin
t-t\cos t}\text{ , }Q^{\prime}=\frac{\cos t+t\sin t}{\sin t-t\cos t}.
\]
The sum
\[
P^{\prime}+Q=\frac{\sin t}{\sin t-t\cos t}%
\]
only vanishes for $t\in\pi\mathbb{Z}$, hence formulas (\ref{P''}),
(\ref{L''}), and (\ref{Q''}) in Proposition \ref{prop1} yield, for $t\notin%
\pi\mathbb{Z}$, $U_{t}^{H_{0}}U_{t}^{H_{1}\circ S_{t}}=\widehat{S}%
_{W^{\prime\prime},m^{\prime\prime}}$ with%
\[
P^{\prime\prime}=Q^{\prime\prime}=\frac{\cos t}{\sin t}\text{ , }%
L^{\prime\prime}=\frac{1}{\sin t}.
\]
We thus have $U_{t}^{H_{0}}U_{t}^{H_{1}\circ S_{t}}=U_{t}^{H}$.

\subsubsection{Inhomogeneous quadratic Hamiltonians}

The results above allow us to prove the following extension of Propositions
\ref{prop8} and \ref{prop9} to inhomogeneous quadratic Hamiltonians, thus
solving a problem unsuccessfully addressed in Burdet \textit{et al}.
\cite{bu}, \S 6. We first remark that every in every such Hamiltonian%
\[
H(z,t)=\frac{1}{2}M(t)z^{2}+m(t)\cdot z
\]
with $M(t)\in\operatorname*{Sym}(2n,\mathbb{R})$ and $z_{0}(t)\in
\mathbb{R}^{2n}$, both being $C^{j}$ ($j\geq2$) functions of $t\in\mathbb{R}$,
the inhomogeneous term can be written as a Hamiltonian of the type
$\sigma(z,\dot{z}(t))$ where
\[
z(t)=J\int_{0}^{t}m(t^{\prime})dt^{\prime}+z(0).
\]

\begin{proposition}
\label{prop11}Let $H$ be a\ Hamiltonian function of the type%
\begin{equation}
H(z,t)=\frac{1}{2}M(t)z^{2}+\sigma(z,\dot{z}(t)). \label{himp}%
\end{equation}
(i) The solution of the corresponding Schr\"{o}dinger equation%
\[
i\hbar\frac{\partial\psi}{\partial t}=\widehat{H}\psi\ \ \text{,}%
\ \ \psi(\cdot,0)=\psi_{0}%
\]
is given by $\psi=U_{t}^{H}\psi_{0}$ with%
\begin{equation}
U_{t}^{H}=e^{\frac{i}{\hbar}\chi(t)}\widehat{S_{t}}\widehat{T}(u(t)\text{ \ ,
\ }u(t)=\int_{0}^{t}S_{t^{\prime}}^{-1}\dot{z}(t^{\prime})dt^{\prime}
\label{imp}%
\end{equation}
where $(\widehat{S_{t}})$ is the evolution operator for the Weyl quantization
of $H_{0}(z,t)=\frac{1}{2}M(t)z\cdot z$ and the phase $\chi$ is given by the
formula%
\begin{equation}
\chi(t)=-\frac{1}{2}\int_{0}^{t}\sigma(S_{t^{\prime}}z(t^{\prime}),\dot
{z}(t^{\prime}))dt^{\prime}. \label{imphase}%
\end{equation}
(ii) If $\psi_{0}\in S_{0}(\mathbb{R}^{n})$ (the Feichtinger algebra), then
$\psi(\cdot,t)\in S_{0}(\mathbb{R}^{n})$ for every $t\in\mathbb{R}$.
\end{proposition}

\begin{proof}
(i) Write $H=H_{0}+H_{1}$ where $H_{0}(z,t)=\frac{1}{2}M(t)z\cdot z$ and
$H_{1}(z,t)=\sigma(z,\dot{z}(t))$. In view of formula (\ref{ut}) in
Proposition \ref{propho} we have $U_{t}^{H_{0}+H_{1}}=\widehat{S}_{t}%
U_{t}^{H_{1}\circ S_{t}}$ where $\widehat{S}_{t}$ is the evolution operator
for Schr\"{o}dinger's equation with Hamiltonian operator $\widehat{H}_{0}$ and
$S_{t}=\Pi(\widehat{S}_{t})$ (Proposition \ref{prop9}). On the other hand, we
have
\[
(H_{1}\circ S_{t})(z,t)=\sigma(S_{t}z,\dot{z}(t))=\sigma(z,S_{t}^{-1}\dot
{z}(t))
\]
hence
\[
\widehat{H_{1}\circ S_{t}}=\sigma(\widehat{z},S_{t}^{-1}\dot{z}(t))=\sigma
(\widehat{z},\dot{u}(t))
\]
where we have set $u(t)=\int_{0}^{t}S_{t^{\prime}}^{-1}\dot{z}(t^{\prime
})dt^{\prime}$. It follows from Proposition \ref{prop10}, formula
(\ref{schrimp1}), that we have
\begin{align*}
U_{t}^{H_{1}\circ S_{t}}  &  =\exp\left(  -\frac{i}{2\hbar}\int_{0}^{t}%
\sigma(z(t^{\prime}),\dot{u}(t))dt^{\prime}\right)  \widehat{T}(u(t))\\
&  =\exp\left(  -\frac{i}{2\hbar}\int_{0}^{t}\sigma(S_{t^{\prime}}z(t^{\prime
}),\dot{z}(t))dt^{\prime}\right)  \widehat{T}(u(t))
\end{align*}
and hence%
\[
U_{t}^{H}=\exp\left(  -\frac{i}{2\hbar}\int_{0}^{t}\sigma(S_{t^{\prime}%
}z(t^{\prime}),\dot{z}(t))dt^{\prime}\right)  \widehat{S}_{t}\widehat{T}%
(u(t))
\]
which is (\ref{imp}). (ii) Follows from formula (\ref{imp}) using Proposition
\ref{prop8} and observing that, up to the unimodular factor $e^{i\chi
(t)/\hbar}$ the propagator $U_{t}^{H}$ is in the inhomogeneous metaplectic
group $\operatorname*{IMp}(n)$.
\end{proof}

Part (ii) of the Proposition above can be interpreted by saying that the
\textquotedblleft phase space concentration\textquotedblright\ of an initial
wavepacket is preserved when the quantum propagator arises from a (homogeneous
or inhomogeneous) quadratic Hamiltonian; we will discuss an extension of this
result in next Section.

\subsection{An extension of the metaplectic representation}

In the 1980s Weinstein \cite{we85} suggested a way to extend the metaplectic
representation by constructing solutions to Schr\"{o}dinger's equations with
non-quadratic potentials. Weinstein's \textquotedblleft generalized
metaplectic operators\textquotedblright\ have recently been rediscovered in
\cite{cogrniro14} where the authors define and study an algebra of Fourier
integral operators and their action on modulation spaces. They assume that the
Hamiltonian is a quadratic function perturbed by a symbol belonging to certain
classes of modulation spaces, the so-called Sj\"{o}strand classes
\cite{sj94,sj95}.

\subsubsection{The Sj\"{o}strand class}

The next application is a regularity result extending (ii) in Proposition
\ref{prop11} hereabove. Let us first introduce a convenient class of symbols
(Sj\"{o}strand \cite{sj94,sj95}, Gr\"{o}chenig \cite{gr01,gr06-2}):

\begin{definition}
The Sj\"{o}strand class $M^{\infty,1}(\mathbb{R}^{2n})$ is the complex vector
space consisting of all $a\in\mathcal{S}^{\prime}(\mathbb{R}^{2n})$ such that%
\[
\sup_{z\in\mathbb{R}^{2n}}|W(a,b)(z,\cdot)|\in L^{1}(\mathbb{R}^{2n})
\]
for all $b\in\mathcal{S}(\mathbb{R}^{2n})$; here $(z,\varsigma)\longmapsto
W(a,b)(z,\varsigma)$ is the cross-Wigner transform on $\mathbb{R}^{2n}%
\times\mathbb{R}^{2n}$.
\end{definition}

The Sj\"{o}strand class is a Banach space for the topology defined by the
equivalent norms $||\cdot||_{b,M^{\infty,1}}$ defined, for $\Phi\in
\mathcal{S}(\mathbb{R}^{2n})$, by%
\[
||a||_{\Phi,M^{\infty,1}}=\int_{\mathbb{R}^{2n}}\sup_{z\in\mathbb{R}^{2n}%
}|W(a,\Phi)(z,\zeta)|d\zeta;
\]
the Schwartz space $\mathcal{S}(\mathbb{R}^{2n})$ is dense in $M^{\infty
,1}(\mathbb{R}^{2n})$. We also mention that $M^{\infty,1}(\mathbb{R}^{2n})$ is
invariant under any linear change of variables: if $a\in M^{\infty
,1}(\mathbb{R}^{2n})$ and $M\in\operatorname*{GL}(2n,\mathbb{R})$, then
$a\circ M\in M^{\infty,1}(\mathbb{R}^{2n})$.

In our context, the interest of this symbol space comes from the two following properties:

\begin{proposition}
\label{prop12}Let $a\in M^{\infty,1}(\mathbb{R}^{2n})$. (i) The Weyl operator
$\widehat{A}=\operatorname*{Op}^{w}(a)$ is bounded on the Feichtinger algebra
$S_{0}(\mathbb{R}^{n})$. (ii) $M^{\infty,1}(\mathbb{R}^{2n})$ is a Banach
algebra for the star-product $\ast_{\hbar}$: if $b\in M^{\infty,1}%
(\mathbb{R}^{2n})$ then $a\ast_{\hbar}b\in M^{\infty,1}(\mathbb{R}^{2n})$, and
for every $\Phi\in\mathcal{S}(\mathbb{R}^{2n})$ there exists a constant
$C_{\Phi}>0$ such that
\begin{equation}
||a\ast_{\hbar}b||_{\Phi,M^{\infty,1}}\leq C_{\Phi}||a||_{\Phi,M^{\infty,1}%
}||b||_{\Phi,M^{\infty,1}} \label{est}%
\end{equation}
for all $a\ast_{\hbar}b\in M^{\infty,1}(\mathbb{R}^{2n})$.
\end{proposition}

We refer to Gr\"{o}chenig \cite{gr01,gr06-2} for a proof of these properties.

\subsubsection{A perturbation of the metaplectic representation}

In a recent paper Cordero \textit{et al}. \cite{cogrniro14} establish the
following result, which considerably extends Proposition \ref{prop8}:

\begin{proposition}
\label{prop13}Let $a\in M^{\infty,1}(\mathbb{R}^{2n})$ and set
\begin{equation}
H(z,t)=\frac{1}{2}M(t)z^{2}+a(z,t).
\end{equation}
Consider the associated Schr\"{o}dinger equation
\[
i\hbar\frac{\partial\psi}{\partial t}=\widehat{H}\psi\ \text{,}\ \psi
(\cdot,0)=\psi_{0}.
\]
We have $\psi(\cdot,t)\in S_{0}(\mathbb{R}^{n})$ for every $t\in\mathbb{R}$ if
and only if $\psi_{0}\in S_{0}(\mathbb{R}^{n})$.
\end{proposition}

This result is very interesting, because it shows that the perturbation of a
quadratic Hamiltonian by an operator associated to a Sj\"{o}strand symbol does
not affect the phase-space concentration of the solution to Schr\"{o}dinger's
equation. Such perturbations can thus be viewed as \textquotedblleft
small\textquotedblright\ from the phase space perspective. Wether a similar
result holds for larger classes of \textquotedblleft
perturbations\textquotedblright\ is an open problem.

It turns out that Proposition \ref{prop13} can be obtained from the
factorization result (\ref{ut}) in Proposition \ref{propho}. We are going to
sketch the argument here; the full details will be published elsewhere. Let us
write $H=H_{0}+H_{1}$ where $H_{0}=\frac{1}{2}M(t)z^{2}$ and $H_{1}=a$; in
view of (\ref{ut}) we have $U_{t}^{H}=\widehat{S}_{t}U_{t}^{a\circ S_{t}}$
where $(\widehat{S}_{t})_{t}$ is the metaplectic isotopy determined by
$\widehat{H}_{0}$. In view of the closedness of $S_{0}(\mathbb{R}^{n})$ under
the action of $\operatorname*{Mp}(n)$ it is thus sufficient to show that
$U_{t}^{a\circ S_{t}}$ maps $S_{0}(\mathbb{R}^{n})$ into itself. Since the
Sj\"{o}strand class is invariant under any linear change of variables, the
proof of Proposition \ref{prop13} thus boils down to the simpler statement:

\begin{lemma}
Let $H\in M^{\infty,1}(\mathbb{R}^{2n})$. Then $U_{t}^{H}:S_{0}(\mathbb{R}%
^{n})\longrightarrow S_{0}(\mathbb{R}^{n})$ for every $t\in\mathbb{R}$.
\end{lemma}

\noindent\textbf{Sketch of the proof}. For simplicity, we assume that $H$ is
time-independent; then $U_{t}^{H}=e^{-it\widehat{H}/\hbar}$; one shows that
the Weyl symbol of $U_{t}^{H}$ can be written as a series
\[
u_{t}=I+H+\frac{1}{2!}H\ast_{\hbar}H+\frac{1}{3!}(H\ast_{\hbar}H\ast_{\hbar
}H)^{3}+\cdot\cdot\cdot
\]
which is convergent in view of the inequality (\ref{est}) in Proposition
\ref{prop12}. Hence $u_{t}\in M^{\infty,1}(\mathbb{R}^{2n})$, and it then
suffices to use part (i) of the same proposition.

\section{Discussion}

\paragraph{Phase space picture of quantum mechanics}

There is dissymmetry between Hamiltonian mechanics and quantum mechanics in
its usual formulation (the Schr\"{o}dinger representation, which we have been
considering here). Hamilton's equations are expressed in terms of a phase
space, endowed with a symplectic structure, while Schr\"{o}dinger's equation
leads to evolution operators defined on functions on configuration space. This
difficulty is actually not a real one, because one can replace the standard
Schr\"{o}dinger equation by the phase space equation%
\[
i\hbar\frac{\partial\Psi}{\partial t}=H\left(  x+\tfrac{1}{2}i\hbar
\partial_{p},p-\tfrac{1}{2}i\hbar\partial_{x}\right)  \Psi
\]
obtained from the Hamilton function $H$ using Bopp quantization instead of
Weyl quantization (see \cite{Birkbis,golu09bis,golu11}. This equation allows
to define a propagator $\widetilde{U}_{t}^{H}$ acting on functions defined on
phase space, which is related to the usual propagator $U_{t}^{H}$ by an
intertwining relation. The phase space Schr\"{o}dinger equation is actually a
variant of deformation quantization; this approach might actually be, both
mathematically and physically, more fruitful than the traditional
Schr\"{o}dinger approach, since it views quantum mechanics as a deformation of
classical mechanics.

\paragraph{Epistemic considerations}

One cannot help being struck by the beauty and harmony of the structures
underlying Hamiltonian flows and isotopies, and their operator counterparts,
quantum isotopies. Both are governed by almost similar -- one is tempted to
say \emph{identical} -- algebraic laws corresponding to underlying physical
processes that are well illustrated by Bohm and Hiley's \textquotedblleft
implicate order\textquotedblright\ \cite{bohi93}. As hinted at in the
Introduction (also see the discussion in Section \ref{seclift}), one actually
passes from Hamiltonian flows to their quantum counterparts by
\textquotedblleft lifting\textquotedblright\ the first to a covering group;
since the most obvious occurrence of this property can be seen on the
metaplectic representation, we have called elsewhere \cite{ICP,gohi15}
\textit{metatron} the entity whose evolution is governed by the quantum motion
induced by the quantum propagator. In that philosophical view, Hamiltonian
flows can be seen as \textquotedblleft unfolding\textquotedblright%
\ reality;\ Hiley's work \cite{Basil1} might well help shed some light on the
deeper algebraic structures; in Hiley's language our product formulas could be
interpreted as algebraic unfolding processes underlying a \textquotedblleft
holomovement\textquotedblright\ (also see de Gosson and Hiley \cite{gohi13}
for a musical comparison). There is however a fundamental difference between
Hamiltonian and quantum mechanics, already visible at the level of
mathematics: the first is local, while the second is non-local, global
(\textit{cf}. the Moyal product (\ref{Moyal1}) where one integrates over the
whole space at time $t$). These issues of classical versus quantum processes
will be further discussed in \cite{gohi15}.

\begin{acknowledgement}
The author has been financed by the Austrian Research Agency FWF; project
number P 23902-N13.
\end{acknowledgement}

\begin{acknowledgement}
I thank Basil Hiley for critical comments about the physical content of this
paper. I would also like to express my warmest thanks to Glen Dennis for
careful proofreading of an early version of the manuscript; needless to say, I
am solely responsible for remaining misprints or typos.
\end{acknowledgement}


\begin{thebibliography}{99}                                                                                               %


\bibitem {AM}R. Abraham, J. E. Marsden. Foundations of Mechanics. Second
Edition, revised, enlarged, and reset, Addison--Wesley Publishing Company,
Inc., Redwood City, CA, 1987

\bibitem {AMR}R. Abraham, J. E. Marsden, T. Ratiu. Manifolds, Tensor Analysis,
and Applications. Applied Mathematical Sciences 75, Springer, 1988

\bibitem {Arnold}V. I. Arnold. Mathematical Methods of Classical Mechanics.
Graduate Texts in Mathematics, second edition, Springer-Verlag, 1989

\bibitem {Banyaga}A. Banyaga. Sur la structure du groupe des
diff\'{e}omorphismes qui pr\'{e}servent une forme symplectique, Comm. Math.
Helv. 53, 174--227 (1978)

\bibitem {Binz}E. Binz, S. Pods. The geometry of Heisenberg groups: with
applications in signal theory, optics, quantization, and field quantization
(Vol. 151). American Mathematical Soc., 2008

\bibitem {Berry}M. V. Berry. Quantal phase factors accompanying adiabatic
changes. Proc. Roy. Soc. London Ser. A, Mathematical and Physical Sciences
392(1802), 45--57 (1984)

\bibitem {beshu}F. A. Berezin, M. A. Shubin. The Schr\"{o}dinger Equation.
Vol. 66. Springer, 1991

\bibitem {bohi93}D. Bohm, B. J. Hiley. The Undivided Universe, London:
Routledge, 1993

\bibitem {Boothby}W. M. Boothby. Transitivity of the automorphisms of certain
geometric structures. Trans. Amer. Math. Soc., 93--100 (1969)

\bibitem {bj}M. Born, P. Jordan. Zur Quantenmechanik, Z. Physik 34, 858--888
(1925) [English translation in: M. Jammer The Conceptual Development of
Quantum Mechanics. New York: McGraw-Hill (1966); 2nd ed: New York: American
Institute of Physics (1989)]

\bibitem {bu}G. Burdet, M. Perrin, M. Perroud. Generating functions for the
affine symplectic group. Comm. Math. Phys., 58(3), 241--254 (1978)

\bibitem {chorin}A. J. Chorin, T. J. Hughes, M. F. McCracken, J. E. Marsden.
Product formulas and numerical algorithms. Comm. Pure Appl. Math., 31(2),
205--256 (1978)

\bibitem {cogrniro14}E. Cordero, K. Gr\"{o}chenig, F. Nicola, L. Rodino.
Generalized metaplectic operators and the Schr\"{o}dinger equation with a
potential in the Sj\"{o}strand class. J. Math. Phys. 55, 081506 (2014)

\bibitem {Dragt}A. J. Dragt. Lie Methods for Nonlinear Dynamics with
Applications to Accelerator Physics. Online book: http://www.physics.umd.edu/dsat/docs/Book8Nov2014.pdf

\bibitem {eg69}Y. V. Egorov. The canonical transformations of
pseudodifferential operators. Uspekhi Matematicheskikh Nauk, 24(5), 235--236 (1969)

\bibitem {em83}G. G. Emch. Geometric dequantization and the correspondence
problem. Int. J. Theor. Phys. 22(5), 397--420 (1983)

\bibitem {er11}B. Erdelyi. Symmetries and their Applications in Beam Physics.
Int. J. Pure Appl. Math. 68, 301 (2011)

\bibitem {hanskernel}H. G. Feichtinger. Un espace de Banach de distributions
temp\'{e}r\'{e}es sur les groupes localement compacts ab\'{e}liens, C. R.
Acad. Sci. Paris, S\'{e}rie A--B 290\textbf{\ }(17), A791--A794 (1980)

\bibitem {fe81}H. G. Feichtinger. On a new Segal algebra,\ Monatsh. Math. 92,
269--289 (1981)

\bibitem {fe81bis}H. G. Feichtinger. Banach spaces of distributions of
Wiener's type and interpolation, in Functional Analysis and Approximation
(Oberwohlfach, 1980), Internat. Ser. Numer. Math. 60 Birkh\"{a}user, Basel,
153--165 (1981)

\bibitem {fo89}G. B. Folland. Harmonic analysis in phase space. No. 122.
Princeton university press, 1989

\bibitem {gala14}W. Gaim, C. Lasser. Corrections to Wigner type phase space
methods. arXiv preprint arXiv:1403.2839 (2014)

\bibitem {AIF}M. de Gosson. Maslov Indices on $\operatorname*{Mp}(n)$.\ Ann.
Inst. Fourier, Grenoble, 40(3), 537--555 (1990)

\bibitem {JMPA}M. de Gosson. The structure of $q$-symplectic geometry. J.
Math. Pures et Appl. 71, 429--453 (1992)

\bibitem {Wiley}M. de Gosson. Maslov classes, metaplectic representation and
Lagrangian quantization. Mathematical research, Vol. 95, Akademie Verlag (1997)

\bibitem {ICP}M. de Gosson. The Principles of Newtonian and Quantum Mechanics:
the Need for Planck's Constant $h$; with a foreword by B. Hiley. Imperial
College Press, 2001

\bibitem {go05}M. de Gosson. On the Weyl representation of metaplectic
operators. Lett. Math. Phys. 72(2), 129--142 (2005)

\bibitem {Birk}M. de Gosson Symplectic Geometry and Quantum Mechanics,
Birkh\"{a}user, Basel, series \textquotedblleft Operator Theory: Advances and
Applications\textquotedblright\ (subseries: \textquotedblleft Advances in
Partial Differential Equations\textquotedblright), Vol. 166 (2006)

\bibitem {go07}M. de Gosson. Metaplectic representation, Conley--Zehnder
index, and Weyl calculus on phase space. Rev. Math. Phys. 19(10), 1149--1188 (2007)

\bibitem {go09}M. de Gosson. On the usefulness of an index due to Leray for
studying the intersections of Lagrangian and symplectic paths. Journal de
math\'{e}matiques pures et appliqu\'{e}es 91(6), 598--613 (2009)

\bibitem {FP}M. de Gosson. The Symplectic Camel and the Uncertainty Principle:
The Tip of an Iceberg? Found. Phys. 99, 194--214 (2009)

\bibitem {stat}M. de Gosson. On the Use of Minimum Volume Ellipsoids and
Symplectic Capacities for Studying Classical Uncertainties for Joint
Position--Momentum Measurements. J. Stat. Mech., P11005 (2010)

\bibitem {Birkbis}M. de Gosson. Symplectic Methods in Harmonic Analysis;
Applications to Mathematical Physics, Birkh\"{a}user, 2010

\bibitem {go11}M. de Gosson. A transformation property of the Wigner
distribution under Hamiltonian symplectomorphisms. J. Pseudo-Differ. Oper.
Appl. 2(1) 91--99 (2011)

\bibitem {go13}M. de Gosson. Symplectic covariance properties for Shubin and
Born--Jordan pseudo-differential operators. Trans. Amer. Math. Soc. 365(6),
3287--3307 (2013)

\bibitem {go14}M. de Gosson. Born--Jordan Quantization and the Equivalence of
the Schr\"{o}dinger and Heisenberg Pictures. Found. Phys. 44(10), 1096--1106 (2014)

\bibitem {gogo03}M. de Gosson, S. de Gosson. Symplectic path intersections and
the Leray index . Symplectic path intersections and the Leray index. Partial
differential equations and mathematical physics (Tokyo, 2001), Progr.
Nonlinear Differential Equations Appl. 52, Birkh\"{a}user Boston, Boston, MA (2003)

\bibitem {gogo06}M. de Gosson, S. de Gosson. Extension of the Conley-Zehnder
index, a product formula, and an application to the Weyl representation of
metaplectic operators. J. Math. Phys. 47(12), 123506 (2006)

\bibitem {gohi11}M. de Gosson, B. Hiley. Imprints of the quantum world in
classical mechanics. Found. Phys. (41)9, 1415--1436 (2011)

\bibitem {gohi13}M. de Gosson, B. Hiley. Hamiltonian Flows and the
Holomovement." Mind and Matter 11.2 (2013): 205-221.

\bibitem {gohi15}M. de Gosson, B. Hiley. Mathematical and Physical Aspects of
Quantum Processes (with Basil Hiley). To appear in Imperial College Press (2014)

\bibitem {golu09}M. de Gosson, F. Luef. Symplectic capacities and the geometry
of uncertainty: the irruption of symplectic topology in classical and quantum
mechanics. Phys. Reps., 484(5),131--179 (2009)

\bibitem {golu09bis}M. de Gosson, F. Luef. On the usefulness of modulation
spaces in deformation quantization. J. Phys. A: Mathematical and Theoretical,
42(31), 315205 (2009)

\bibitem {golu11}M. de Gosson, F. Luef. Preferred quantization rules:
Born--Jordan vs. Weyl; applications to phase space quantization. J.
Pseudo-Differ. Oper. Appl 2(1), 115--139 (2011)

\bibitem {GB}J. M. Gracia-Bondia. Generalized Moyal Quantization on
Homogeneous Symplectic spaces. Contemporary Mathematics 134, Amer. Math. Soc.,
Providence, RI, 93--114, 1992

\bibitem {gr01}K. Gr\"{o}chenig. Foundations of Time-Frequency Analysis, Appl.
Numer. Harmon. Anal., Birkh\"{a}user, Boston, MA (2001)

\bibitem {gr06}K. Gr\"{o}chenig. Composition and spectral invariance of
pseudodifferential operators on modulation spaces. Journal d'analyse
math\'{e}matique, 98(1), 65--82 (2006)

\bibitem {gr06-2}K. Gr\"{o}chenig. Time-frequency analysis of Sj\"{o}strand's
class. Rev. Mat. Iberoamericana 22(2), 703--724 (2006)

\bibitem {Gromov}M. Gromov. Pseudoholomorphic curves in symplectic
manifolds.\ Invent. Math. 82, 307--347 (1985)

\bibitem {gr76}A. Grossmann. Parity operators and quantization of $\delta
$-functions. Commun. Math. Phys. 48, 191--193 (1976)

\bibitem {Hansen}F. Hansen. The Moyal Product and Spectral Theory for a Class
of Infinite Dimensional Matrices. PubL RIMS, Kyoto Univ. 26, 885--933 (1990)

\bibitem {ha66}Y. Hatakeyama. Some notes on the group of automorphisms of
contact and symplectic structures, T\^{o}hoku Math. J. 18, 338--347 (1966)

\bibitem {Basil1}B. J. Hiley. Process, distinction, groupoids and Clifford
algebras: an alternative view of the quantum formalism. New Structures for
Physics. Springer Berlin Heidelberg, 705--752, 2011

\bibitem {Hoermander}L. H\"{o}rmander. The Analysis of Linear Partial
Differential Operators, vol. III, Springer-Verlag, Berlin, 1985

\bibitem {HZ}H. Hofer, E. H. Zehnder. Symplectic Invariants and Hamiltonian
Dynamics., Birkh\"{a}user Advanced texts (Basler Lehrb\"{u}cher),
Birkh\"{a}user Verlag, 1994L.

\bibitem {kela13}J. Keller, C. Lasser. Propagation of Quantum Expectations
with Husimi Functions. SIAM Journal on Applied Mathematics 73(4), 1557--1581 (2014)

\bibitem {Katok}A. Katok. Ergodic perturbations of degenerate integrable
Hamiltonian systems. Izv. Akad. Nauk. SSSR Ser. Mat. 37, 539--576 (1973)
[Russian]. English translation: Math. USSR izvestija 7, 535--571 (1973)

\bibitem {Leray}J. Leray. Lagrangian Analysis and Quantum Mechanics,\ a
mathematical structure related to asymptotic expansions and the Maslov index
(the MIT Press, Cambridge, Mass., 1981); translated from Analyse Lagrangienne
RCP 25, Strasbourg Coll\`{e}ge de France (1976--1977)

\bibitem {Littlejohn}R. G. Littlejohn. The semiclassical evolution of wave
packets, Physics Reports 138(4--5), 193--291 (1986)

\bibitem {ma98}G. W. Mackey. The Relationship Between Classical and Quantum
Mechanics. In Contemporary Mathematics 214, Amer. Math. Soc., Providence, RI, 1998

\bibitem {mivi94}P. W. Michor, C. Vizman. N--transitivity of Certain
Diffeomorphism Groups. Acta Math. Univ. Comenianae 63(2), 221--225 (1994)

\bibitem {mac}R. I. McLachlan, P. Atela. The accuracy of symplectic
integrators. Nonlinearity 5, 541--562 (1992)

\bibitem {Marsden}J. E. Marsden. Hamiltonian one parameter groups a
mathematical exposition of infinite dimensional Hamiltonian Systems with
applications in classical and quantum mechanics. Archive for Rational
Mechanics and Analysis 28(5), 362--396 (1968)

\bibitem {RS}M. Reed, B. Simon. Methods of Modern Mathematical Physics.
Academic Press, New York, 1972

\bibitem {Reiter}H. Reiter. Metaplectic groups and harmonic analysis. Springer
Berlin Heidelberg, 1988

\bibitem {ro77}A. Royer. Wigner functions as the expectation value of a parity
operator. Phys. Rev. A 15, 449--450 (1977)

\bibitem {sj94}J. Sj\"{o}strand. An algebra of pseudodifferential operators,
Math. Res. Lett. 1(2) 185--192 (1994)

\bibitem {sj95}J. Sj\"{o}strand. Wiener type algebras of pseudodifferential
operators, S\'{e}minaire sur les \'{E}quations aux D\'{e}riv\'{e}es
Partielles, 1994--1995, \'{E}cole Polytech., Palaiseau, Exp. No. IV, 21 (1995)

\bibitem {sh87}M. A. Shubin. Pseudodifferential Operators and Spectral Theory,
Springer-Verlag, 1987 [original Russian edition in Nauka, Moskva, 1978]

\bibitem {Souriau}J.-M. Souriau. Construction explicite de l'indice de Maslov.
Applications. Group Theoretical Methods in Physics, 117--148 (1976)

\bibitem {spanier}E. Spanier. Algebraic Topology. McGraw-Hill, 1966

\bibitem {Stone}M. H. Stone. On one-parameter unitary groups in Hilbert Space.
Ann. Math. 33(3), 643--648 (1932)

\bibitem {Wang}D. Wang. Some aspects of Hamiltonian systems and symplectic
algorithms, Physica D. 73, 1--16 (1994)

\bibitem {we85}A. Weinstein. A symbol class for some Schr\"{o}dinger equations
on $\mathbb{R}^{n}$. Am. J. Math. 107(1), 1--21 (1985)

\bibitem {ya87}K. Yajima. Existence of solutions for Schr\"{o}dinger evolution
equations. Comm. Math. Phys. 110(3), 415--426 (1987)

\bibitem {ya96}K. Yajima. Smoothness and Non-Smoothness of the Fundamental
Solution of Time Dependent Schr\"{o}dinger Equations. Comm. Math. Phys. 181,
605--629 (1996)

\bibitem {zehnder}E. Zehnder. Lectures on Dynamical Systems: Hamiltonian
Vector Fields and Symplectic Capacities. EMS Textbooks in Mathematics.
European Mathematical Society, 2010.
\end{thebibliography}
\end{document}